\documentclass[a4paper]{article}

\usepackage[english,activeacute]{babel}
\usepackage[utf8x]{inputenc}
\usepackage[T1]{fontenc}

\usepackage{lipsum}
\usepackage{amsfonts}
\usepackage{graphicx}
\usepackage{epstopdf}
\usepackage{algorithmic}

\usepackage{booktabs}
\usepackage{multirow}
\usepackage{bbm, dsfont}
\usepackage{color}
\usepackage{boldline}

\usepackage{amssymb}

\usepackage{amsmath,amsthm}
\usepackage{amssymb}
\usepackage{bbm, dsfont}
\usepackage{wasysym}%
\usepackage{color}%
\usepackage{capt-of}

\usepackage[a4paper,top=3cm,bottom=2cm,left=3cm,right=3cm,marginparwidth=1.75cm]{geometry}

\usepackage{amsmath}
\usepackage{graphicx}
\usepackage[colorinlistoftodos]{todonotes}
\usepackage[colorlinks=true, allcolors=blue]{hyperref}

\newtheorem{proposition}{Proposition}
\newtheorem{assumption}{Assumption} 
\newtheorem{lemma}{Lemma}
\newtheorem{corollary}{Corollary}
\newtheorem{remark}{Remark}

\newcommand{\R}{{\mathbb R}} 

\newcommand{\C}{{\mathcal C}} 

\newcommand{\rstar}{r^{\star}}
\newcommand{\xeq}{x^{\star}}
\newcommand{\xeqa}{x_1^{\star}}
\newcommand{\xeqb}{x_2^{\star}}

\newcommand{\omegaeq}{\omega^{\star}}
\newcommand{\muini}{\mu^{{\rm in},i}}
\newcommand{\muouti}{\mu^{{\rm out},i}}
\newcommand{\indi}[1]{\mathbbm{1}_{#1}}

\usepackage{color}

\title{Inmate population models with nonhomogeneous sentence lengths and their effects in an epidemiological model}

\author{{\sc P. Gajardo$^{1}$ and V. Riquelme$^{1}$}\\[2mm]
$^{1}$ Departamento de Matem\'atica, Universidad T\'ecnica Federico Santa Mar\'ia,\\Avenida Espa\~na 1680, Valpara\'iso, Chile\\[2mm]
{\tt pedro.gajardo@usm.cl, victor.riquelmef@usm.cl}\\[2mm]
}

\begin{document}
\maketitle

\begin{abstract}
In this work, we develop an inmate population model with a sentencing length structure. The sentence length structure of new inmates represents the problem data and can usually be estimated from the histograms corresponding to the conviction times that are sentenced in a given population. 
We obtain a transport equation, typically known as the McKendrick equation, the homogenous version of which is included in population models with age structures. Using this equation, we compute the inmate population and entry/exit rates in equilibrium, which are the values to consider in the design of a penitentiary system. With data from the Chilean penitentiary system, we illustrate how to perform these computations. In classifying the inmate population into two groups of sentence lengths (short and long), we incorporate the SIS (susceptible-infected-susceptible) epidemiological model, which considers the entry of infective individuals. We show that a failure to consider the structure of the sentence lengths---as is common in epidemiological models developed for inmate populations---for prevalences of new inmates below a certain threshold induces an underestimation of the prevalence in the prison population at steady state. The threshold depends on the basic reproduction number associated with the nonstructured SIS model with no entry of new inmates. We illustrate our findings with analytical and numerical examples for different distributions of sentencing lengths.

\end{abstract}

\section{Introduction}

Currently, several communicable diseases, such as sexually transmitted infections (STIs), remain a public health problem that is far from being controlled \cite{world2016}. Generally, prisons present far higher prevalences of certain diseases than the general population. This is due, among many factors, to the existence of crowded environments, high-risk behaviors such as unprotected sexual relations \cite{kouyoumdjian2012}, and the increase in the probability of the appearance of disease risk factors, such as depression and drug use, during imprisonment. Another relevant factor in penitentiary systems is related to the deficiencies of prison health systems, which imply the existence of barriers to access to care, delayed diagnoses and prolonged contagion times \cite{belenko2008, khan2009incarceration, MariaetAl2011,kouyoumdjian2012}. This health problem is not only a penitentiary concern but also a general social issue because prisons act as reservoirs for diseases, which are later transmitted to the community when inmates are released or come into contact with the outside population, such as visitors and prison workers. For the aforementioned reasons, the World Health Organization (WHO) includes prisoners among the key populations that should be the focus of interventions designed to reduce the burden of diseases \cite{who2007,world2016}.

Mathematical models of communicable diseases in penitentiary systems can be found in the recent literature \cite{Beauparlant:2016,Brauer:2001,castillo2020,Legrand:2008,Witbooi:2017}. The main objective of these models is to develop and assess control strategies regarding the spread of these diseases. The main characteristics of these models and two of the difficulties in conducting a theoretical analysis are the entry or immigration of infected individuals and a variable population size. These features are also present in age-structured epidemiological models with vertical transmission, but as indicated in \cite{Brauer:2001}, vertical transmission models generally include a flow of new infectives proportional to the number of infectives already in the population and thus may have a disease-free equilibrium \cite{Busenberg:2012}. To the best of our knowledge, the first theoretical analysis for this class of models was conducted in \cite{Brauer:2001}, where the SIS (susceptible-infected-susceptible) and SIR (susceptible-infected-recovered) models are studied. 

These previous studies ignore the sentence length structure of inmates and assume the entry (and exit) rate of new inmates to be the inverse of the average sentence length (or residence time in prisons). However, it seems appropriate to consider the sentence length structure of inmates in epidemiological models, since individuals with a longer sentence length are exposed to the disease under study for a longer time.

Under the framework of discrete-time models, we find epidemiological models developed for prison populations that consider the sentence length structure. In \cite{Ayer:2019}, an age-structured operational research model is developed, where the sentence length structure is included for studying hepatitis C treatment strategies in U.S. prisons. In \cite{Gani:1997}, a model (also in discrete time) is introduced for the spread and control of HIV in prison populations. In this work, the authors consider two classes of sentencing lengths (short and long) and conclude, from numerical simulations, that the epidemic is not very sensitive to the length of sentences in the range of parameters used in their simulations.

In this paper, we first introduce an inmate population model with a sentencing length structure. We obtain a transport equation in the form of a partial differential equation, typically known as the McKendrick equation, the homogenous version of which is included in population models with an age structure \cite{BrauerCastillo:2012}. This equation allows us to compute the inmate population and the entry/exit rates at steady state, which are values to consider in the design of a penitentiary system. Our main objective is not to solve or analyze the obtained McKendrick equation, which can be achieved using the method of characteristics. Instead of coupling the McKendrick equation with an epidemiological model (which for homogeneous equations ---obtained from age-structured population models--- is analyzed in \cite{Akushevich:2019,BrauerCastillo:2012,Kuniya:2015}), we divide the inmate population into two classes depending on their sentencing lengths (short and long), and in the obtained model, we incorporate the SIS epidemiological model. This epidemiological model is compared, in equilibrium, with the model obtained when ignoring the sentencing length structure. We prove that failing to consider the structure of the sentence lengths for prevalences of new inmates\footnote{We assume that the prevalence in new inmates is the steady-state prevalence of the population outside the prison.} below a certain threshold induces an underestimation of the prevalence in the prison population at steady state. The threshold depends on the basic reproduction number associated with the nonstructured SIS model with no entry of new inmates.

The remainder of the paper is organized as follows: In Section \ref{sec:modelo}, we develop a model for an inmate population that considers an nonhomogeneous distribution of the sentence lengths of new inmates and the removal rate from the prison. We perform the analysis at steady state for the model, and we obtain expressions for the population distribution (with respect to the remaining sentence time) and the population size with respect to the data of the problem. Using real data from a prison in Chile, in Section \ref{subsec:examples} we illustrate how to obtain these key values. In Section \ref{sec:SIS}, we study the effect of considering one or two classes of initial sentencing lengths in the estimation of the proportion of infected people in an epidemiological population model in equilibrium. We provide conditions under which the homogeneous model (single class without a sentencing length structure) underestimates the number of infective individuals relative to the two-class model. Finally, in Section \ref{subsec:simulations}, we illustrate our findings with two numerical examples for different distributions of initial sentencing lengths and then conclude with some final remarks in Section \ref{sec:conclusion}.

\section{Inmate population models with sentence length structure}\label{sec:modelo}

\subsection{Population model with inhomogeneous sentence time structure}

In this section, we develop a population model for a prison population consisting of a continuous-state, continuous-time model for the number of inmates inside the prison. Suppose that at time $t$, prisoners enter the prison at a rate $\lambda_t\geq0$ (the number of people per unit of time) corresponding to inmates newly sentenced, or transferred from other institutions, with a distribution of initial sentence length $q_t(dr)$ that, for each $t\geq0$, is a nonnegative measure on $\mathcal B(\C)$, the Borel $\sigma-$algebra on the set of sentence times $\C=(0,\infty)$, with total mass 1 (that is, a probability measure). Consider the existence of a removal rate of prisoners $d_t\geq0$ corresponding to transfer of inmates to other institutions, pardon or commutation of sentences, death, etc., in units of time$^{-1}$, which is independent of their respective remaining sentence time in prison.

Throughout this work, we make the following assumption regarding the distribution of initial sentence lengths $q_t(\cdot)$:

\begin{assumption}\label{ass:measures}
For all $t \ge 0$, the distribution $q_t(\cdot)$ is a mixture measure of a continuous distribution and a discrete distribution with finite support. That is, it can be decomposed as
\begin{equation*}
q_t(dr) = q_{a,t}(dr) + \sum_{r_t^j\in \mathcal D_t} p_t^j \delta_{r_t^j}(dr),
\end{equation*}
where, for all $t \ge 0$, one has the following:
\begin{itemize}\itemsep0em
\item $q_{a,t}(\cdot)$ is a nonnegative measure, absolutely continuous with respect to the Lebesgue measure on $\R$; that is, there exists a density function $q_{a,t}'(\cdot)$ that is Lebesgue measurable, such that $q_{a,t}(I)=\int_I q_{a,t}'(r)dr$ for all $I\in\mathcal B(\C)$.
\item $\mathcal D_t\subseteq\mathcal C$ is a finite (possibly empty) set.
\item For all $r_t^j\in \mathcal D_t$, $\delta_{r_t^j}(dr)$ denotes the Dirac measure concentrated on $\{r_t^j\}$.
\item For all $r_t^j\in \mathcal D_t$, $p_t^j>0$.
\end{itemize}

Moreover, we assume that for all $t\geq0$, the mean value of the distribution $q_t(\cdot)$, denoted by $T_{q_t}=\int_{(0,\infty)}rq_t(dr)$, is finite.
\end{assumption}

Under Assumption \ref{ass:measures}, $q_t(\cdot)$ has a generalized density function $q_t'(r)$, in the sense that we can write
\begin{equation*}
q_t'(r) = q_{a,t}'(r) + \sum_{r_t^j\in \mathcal D_t}p_t^j \delta_{r_t^j}'(r),
\end{equation*}
where the Dirac generalized function $\delta_{x}'(r)$ can be regarded as a density for the Dirac measure $\delta_{x}(dr)$, in the sense that
\begin{equation*}
\delta_{x}(I)=\int_I \delta_{x}'(r)dr = \left\{\quad
\begin{split}
1, & \quad x\in I,\\
0, & \quad x\notin I.
\end{split}
\right.
\end{equation*}

In what follows, we do not make distinctions between the absolutely continuous and discrete parts of $q_t(\cdot)$ under the sign of the integral, meaning the following: for every Borel-measurable function $f:\C\rightarrow\R$,
\begin{equation*}
\int_{\C} f(r)q_t(dr) = \int_{\C} f(r)q_{a,t}'(r)dr + \sum_{r_t^j\in \mathcal D_t} f(r_t^j)p_t^j.
\end{equation*}

\begin{remark}\label{rem:Tq}
If $\nu(\cdot):\mathcal B(\R_{+})\rightarrow\R_{+}$ is a nonnegative Borel measure concentrated on the set $\R_{+}$ and $f:\R_+\rightarrow\R$ is a Borel-measurable function, we have
\begin{equation*}
\int_{\R_{+}}f(x)\nu(dx) = \int_0^{\infty} \nu(\{x\in\R_{+}\,|\, f(x)>\alpha\})d\alpha,
\end{equation*}
according to \cite[Theorem 2.9.3]{Bog2007}. Since for all $t\geq 0$, $q_t(\cdot)$ is concentrated on the positive real line, its mean value $T_{q_t}:=\int_0^{\infty}rq_t(dr)$ can be computed as
\begin{equation}\label{eq:Tq}
T_{q_t} = \int_0^{\infty} q_t( (r,\infty) ) dr.
\end{equation}
\end{remark}

Given a set $I\subseteq\C$, let us define $N_t^{I}$ as the number of prisoners at time $t$, whose remaining times in prison belong to the set $I$. Let us suppose that, for each $t$, there exists a population density $\rho_t^r$ (the number of people per unit of time), with respect to the remaining time in prison $r$, such that
\begin{equation*}
N_t^I = \int_I \rho_t^r dr, \quad \forall I\in\mathcal B(\C).
\end{equation*}
Define $E_t^{I}$ and $D_t^{I}$ as the rates of prisoners (number of prisoners per unit of time) that enter the prison and are removed from prison, respectively, at time $t$, whose remaining sentence times belong to $I\in\mathcal B(\C)$.

The following proposition characterizes the population density $\rho_t^r$ as a function of the rates $\lambda_t$ and $d_t$ and the distribution $q_t(\cdot)$.

\begin{proposition}\label{prop:mckendrick}
Suppose that the functions $t\mapsto \lambda_t$ and $t\mapsto d_t$ are continuous and that the family of measures $(q_t(\cdot))_{t\geq0}$ satisfies the following continuity hypothesis:
\begin{equation}\label{hip:medidas_cuerpo}
\forall t\geq0,\forall A\in \mathcal B(\R_+),\forall \varepsilon>0,\exists\eta\in(0,1):~ |z|,|z'|<\eta\Rightarrow |q_{t+z}(A+z')-q_{t}(A+z')|<\varepsilon.
\end{equation}
Let us consider $0\leq s < r$. Then, the number of prisoners whose remaining sentence times belong to the interval $(s,r]$ satisfies the equation
\begin{equation}\label{eq:balance_dt_muerte}
\frac{d}{dt} N_{t}^{(s,r]} = \rho_t^{r} - \rho_t^{s}
+ \lambda_{t} q_{t}((s,r]) - d_t N_t^{(s,r]},\quad \mbox{ a.e. } t\geq0 .
\end{equation}

Therefore, under Assumption \ref{ass:measures}, the density of the prison population $\rho_t^r$ satisfies the McKendrick transport equation
\begin{equation}\label{eq:pde_muerte}
\frac{\partial \rho_t^{r}}{\partial t} + d_t\rho_t^{r} \,  = \frac{\partial\rho_t^{r}}{\partial r} + \lambda_{t} q'_{t}(r), \quad \mbox{ a.e. } t\geq0, r\geq0 ,
\end{equation}
from an initial distribution $(\rho_0^r)_{r\geq0}$.
\end{proposition}

\begin{proof}
We know that
\begin{equation*}
E_t^I = \int_I\lambda_tq_t(dr) = \lambda_tq_t(I)\quad ,\quad D_t^I = d_t\int_I \rho_t^{r}dr = d_t N_t^I.
\end{equation*}

Consider an interval $I=(s,r]$ of remaining sentence times and a time step $\Delta t$. Mass balance analysis yields
\begin{equation}\label{eq:balance_conteo_muerte}
N_{t+\Delta t}^{(s,r]} \,=\, N_{t}^{(s+\Delta t,r+\Delta t]} + \int_0^{\Delta t} E_{t+\tau}^{(s+\Delta t-\tau, r+\Delta t-\tau]}d\tau - \int_0^{\Delta t} D_{t+\tau}^{(s+\Delta t-\tau, r+\Delta t-\tau]}d\tau.
\end{equation}

We note that for $I_1,I_2\in\mathcal B(\C)$ disjoint sets, $N_t^{I_1}+N_t^{I_1}=N_t^{I_1\cup I_2}$. Thus,
\begin{equation*}\label{eq:conteo_ints_disjuntos}
N_{t}^{(s+\Delta t,r+\Delta t]}= N_{t}^{(s,r]}+N_{t}^{(r,r+\Delta t]}-N_{t}^{(s,s + \Delta t]}.
\end{equation*}

Then, \eqref{eq:balance_conteo_muerte} becomes the mass balance equation on the interval $(s,r]$:
\begin{equation}\label{eq:balance_masa_muerte}
\begin{split}
N_{t+\Delta t}^{(s,r]} - N_{t}^{(s,r]} = &\, N_{t}^{(r,r+\Delta t]}-N_{t}^{(s,s+\Delta t]}  + \int_0^{\Delta t} \lambda_{t+\tau} q_{t+\tau}((s+\Delta t-\tau, r+\Delta t-\tau])d\tau\\
&\, - \int_0^{\Delta t} d_{t+\tau}N_{t+\tau}^{(s+\Delta t-\tau,r+\Delta t-\tau]}  d\tau.
\end{split}
\end{equation}
Note that for all $t\geq0$ fixed, $\varphi_t(\cdot)=\lambda_tq_t(\cdot)$ and $\phi_t(\cdot)=d_t N_t^{\mathbf{\cdot}}$ define positive measures that, under the hypotheses of the proposition, satisfy hypothesis \eqref{hip:medidas_cuerpo}. Thus, by Proposition \ref{prop:Apendice} in the Appendix, dividing \eqref{eq:balance_masa_muerte} by $\Delta t$ and taking limits as $\Delta t\searrow0$, we obtain balance equation \eqref{eq:balance_dt_muerte}.

Now, under Assumption \ref{ass:measures}, $q_t(dr)$ has a density $q'_{t}(r)$, that is, for all $I\in\mathcal B(\mathcal C)$, $q(I)=\int_I q'(r)dr$. Since $N_t^{(s,r]}=\int_{(s,r]}\rho_t^{\alpha}d\alpha$, \eqref{eq:balance_dt_muerte} can be rewritten as
\begin{equation}\label{eq:balance_dt_muerte_2}
\int_{(s,r]}\left[\, \frac{\partial \rho_t^{\alpha}}{\partial t} + d_t\rho_t^{\alpha} \, \right] d\alpha \,=\, \rho_t^{r} - \rho_t^{s} + \lambda_{t} \int_{(s,r]}q'_{t}(\alpha)d\alpha,\quad 0\leq s< r.
\end{equation}
Dividing \eqref{eq:balance_dt_muerte_2} by $r-s$ and taking limits as $s\nearrow r$, by \cite[Theorem 5.4.2]{Bog2007}, we obtain the transport equation \eqref{eq:pde_muerte}.
\end{proof}

\begin{remark}
The McKendrick equation appears naturally in population models with an age structure \cite{BrauerCastillo:2012}. Our prison population model is in some sense an age-structured model where, instead of accounting for age, we consider the remaining life of an individual, given a distribution of life span. One of the differences between e\-qua\-tion \eqref{eq:pde_muerte} and that arising in age-structured population models is the nonhomogeneous term $\lambda_{t} q'_{t}(r)$, which appears because the entry of new inmates is not proportional to the current population, contrary to typical age-structured models.
\end{remark}

\subsubsection{Steady-state analysis}\label{subsec:eq_pop}

In this section, we study the behavior of the system in the long term, assuming that $\lambda,d,q(\cdot)$ does not depend on $t$. Regarding the solution of \eqref{eq:pde_muerte} under this assumption, we have the following proposition:
\begin{proposition}\label{prop:param_indep_t}
Suppose that $\lambda,d,q(\cdot)$ does not depend on $t$ and that Assumption \ref{ass:measures} holds with $T_q<\infty$ being the mean value of $q(\cdot)$ and $\mathcal D$ being the support of the discrete part of $q(\cdot)$. If the initial density $(\rho_0^r)_{r\geq0}$ defines a finite measure consisting of a mixture of a continuous measure and a discrete measure with finite support, then:
\begin{enumerate}
\item For all $t$, the solution $(\rho_t^r)_{r\geq0}$ of equation \eqref{eq:pde_muerte} has finite mass. Moreover, for all $t\geq0$, $\rho_t^r$ converges to 0 as $r$ tends to infinity in a set whose complement has a null Lebesgue measure.
\item The population density satisfies \begin{equation}\label{eq:lim_rho_t}
\lim_{t\rightarrow\infty}\rho_t^r \,=\, \lambda \int_{(r,\infty)}e^{-d(\alpha-r)}q(d\alpha),\quad\mbox{ a.e. } r\geq0.
\end{equation}
Moreover, the previous formula holds for all but finitely many $r\geq0$.
\item For all $I\in\mathcal B(\R_{+})$, we have
\begin{equation}\label{eq:N_tI_lim}
\lim_{t\rightarrow\infty} N_t^{I}\,=\,\lambda \int_I \int_{(r,\infty)}e^{-d(\alpha-r)}q(d\alpha) dr.
\end{equation}
\end{enumerate}
\end{proposition}

\begin{proof}
Equation \eqref{eq:pde_muerte} can be explicitly solved via the method of characteristics. 
Integration along the characteristics leads to the possible solutions
\begin{equation}\label{eq:rho_t_r}
\begin{split}
\rho_t^{r^-} := \rho_0^{r+t}e^{-dt} + \lambda\int_{[r,r+t]}e^{-d(\alpha-r)}q(d\alpha),\\[2mm]~ \rho_t^{r^+} := \rho_0^{r+t}e^{-dt} + \lambda\int_{(r,r+t]}e^{-d(\alpha-r)}q(d\alpha).
\end{split}
\end{equation}
Note that the set of points $r\in\mathcal C$ for which $\rho_t^{r^-} \neq \rho_t^{r^+}$ is at most finite. Indeed, it coincides with $\mathcal D$. Thus, the solution $\rho_t^r$ of \eqref{eq:pde_muerte} satisfies, for all $t\geq0$, $\rho_t^r=\rho_t^{r^-}=\rho_t^{r^+}$, a.e. $r>0$ (indeed, for all $r\notin\mathcal D$). Thus,
\begin{equation}\label{eq:bound_rho_t_r}
\begin{split}
\rho_t^{r} \,\leq\, \rho_t^{r^-} \leq&~ \rho_0^{r+t}e^{-dt} + \lambda\left(q(\{r\})+\int_{(r,\infty)}e^{-d(\alpha-r)}q(d\alpha)\right)\\
\leq&~ \rho_0^{r+t}e^{-dt} + \lambda q(\{r\}) + \lambda q((r,\infty)),
\end{split}
\end{equation}
where $q(\{r\})>0$ only for $r\in\mathcal D$, which is a set with a null Lebesgue measure. The right-hand-side expression in \eqref{eq:bound_rho_t_r} is Lebesgue-integrable with respect to $r$.

\begin{enumerate}
\item Integrating \eqref{eq:bound_rho_t_r} with respect to $r$, on the set $(0,\infty)$, we obtain
\begin{equation*}
N_t \leq N_0^{(t,\infty)}e^{-dt} + \lambda \int_0^{\infty} q((r,\infty)) dr.
\end{equation*}
According to Remark \ref{rem:Tq}, $T_q = \int_0^{\infty} q( (r,\infty) ) dr$, which is finite by hypothesis. This, combined with the fact that $N_0^{(t,\infty)}\leq N_0<\infty$, proves the finiteness of the total mass $N_t$, $t\geq0$. 

Since $N_t=\int_{0}^{\infty}\rho_t^rdr$ is finite, then $r\mapsto\rho_t^r$ is an integrable function. Thus, $\rho_t^r$ converges to 0 as $r\nearrow\infty$ in a set whose complement has a null Lebesgue measure.

\item From \eqref{eq:rho_t_r}, we note that for $t$ large enough, the points with an initial positive measure, with respect to $(\rho_0^{\alpha})_{\alpha\geq0}$, are no longer involved in the expressions for $\rho_t^{r^-}$ and $\rho_t^{r^+}$. This, along with the fact that the support of the discrete part of $q(\cdot)$ is finite, implies that for $t$ large enough, $\rho_t^{r^+}=\rho_t^{r^-}$ except at a finite set, which implies that $r\mapsto \rho_t^{r}$ is a function continuous by parts, with at most finitely many discontinuities.

Formula \eqref{eq:lim_rho_t} is directly obtained from \eqref{eq:rho_t_r}.

\item From \eqref{eq:lim_rho_t}, we have the $r-$a.e. pointwise convergence
\begin{equation}\label{eq:densidad_pre_TCD}
\lim_{t\rightarrow\infty}\rho_t^r\mathbbm{1}_{I}(r)=\lambda \int_{(r,\infty)}e^{-d(\alpha-r)}q(d\alpha)\mathbbm{1}_{I}(r),
\end{equation}
where $\mathbbm{1}_{I}(r)$ denotes the indicator function of the set $I$, which means that $\mathbbm{1}_{I}(r)=1$ if $r\in I$, and $\mathbbm{1}_{I}(r)=0$ otherwise. Moreover, from \eqref{eq:bound_rho_t_r}, we obtain for $t_0$ large enough (fixed), the bound
\begin{equation*}
\rho_t^r\mathbbm{1}_{I}(r)\leq \left(\rho_0^{r+t_0} + \lambda q((r,\infty))\right)\mathbbm{1}_{I}(r),
\end{equation*}
which is an integrable function. Integrating \eqref{eq:densidad_pre_TCD}, by the Lebesgue dominated convergence theorem \cite[Theorem 2.8.1]{Bog2007}, we conclude that \eqref{eq:N_tI_lim}.
\end{enumerate}
\end{proof}

\begin{remark}
For the time-dependent case, a similar result to point 1. of Proposition \ref{prop:param_indep_t} can be proved, if the initial density $(\rho_0^r)_{r\geq0}$ satisfies the hypotheses of Proposition \ref{prop:param_indep_t}, $(\lambda_t)_{t\geq0}$ is bounded, and the family $(q_t(\cdot))_{t\geq0}$ is uniformly tight \cite{cohn2013} (or does not depend on $t$), with every $q_t(\cdot)$ satisfying Assumption \ref{ass:measures}, under the hypotheses of Proposition \ref{prop:mckendrick}. Then, if in \eqref{eq:balance_dt_muerte} we replace the particular value $s=0$ and take limits as $r\nearrow\infty$, we obtain the ODE associated with the total prison population with positive remaining time $N_t := N_t^{(0,\infty)}$:
\begin{equation}\label{eq:balance_poblacion}
\frac{d}{dt} N_{t} = \lambda_{t} - \rho_t^{0} - d_t N_t, \quad \mbox{ a.e. } t\geq0,
\end{equation}
with $\rho_t^r$ being the solution of \eqref{eq:pde_muerte}. The quantity $\rho^0_t$ corresponds to the exit rate of the prison (the number of people per unit of time), which is interpreted as the rate of inmates that have effectively completed their sentences at instant $t$.
\end{remark}

Proposition \ref{prop:param_indep_t} shows that the density $\rho_t^r$ converges as $t\rightarrow\infty$ for all but finitely many $r\geq0$. Assume now that the parameters $\lambda,d,q(\cdot)$ are time independent and that the system is in a stationary regime. Denote the stationary density $\rho^r$ and the stationary population with remaining sentence time in the interval $I\in\mathcal B(\mathcal C)$ by $N^{I}$. The mass-balance equation \eqref{eq:balance_dt_muerte} in equilibrium states that $\rho^r$ is the solution of the stationary equation
\begin{equation}\label{eq:rhos_muerte_001}
\rho^{r} - \rho^{s} + \lambda q((s,r]) \,=\, d N^{(s,r]}, \quad \mbox{ for all } 0\leq s< r.
\end{equation}

Let us define $N:=N^{(0,\infty)} = \int_0^{\infty}\rho^{\alpha}d\alpha$ as the total number of prisoners in equilibrium. In what follows, for a function $f:[0,\infty)\rightarrow\R$, we denote its Laplace transform by $\mathcal L(f)(\xi)=\int_0^{\infty}e^{-\xi x}f(x)dx$, while for a signed measure $\nu:\mathcal B([0,\infty))\rightarrow\R\cup\{+\infty\}$, we denote its Laplace transform by $\hat \nu(\xi)=\int_{[0, \infty)}e^{-\xi x}\nu(dx)$.

In the following proposition, we show that the limiting density \eqref{eq:lim_rho_t} obtained in point 2 of Proposition \ref{prop:param_indep_t} coincides with the unique integrable solution of \eqref{eq:rhos_muerte_001}.

\begin{proposition}\label{prop:densidad_limite}
Suppose that Assumption \ref{ass:measures} holds. Let $\rho^r$ be a solution of \eqref{eq:rhos_muerte_001} with finite mass, that is, such that $N=\int_0^{\infty}\rho^rdr<\infty$. Then,
\begin{equation}\label{eq:comp_densidades_rho_limit}
\rho^r  = \lambda \int_{(r,\infty)}e^{-d(\alpha-r)}q(d\alpha), \quad \mbox{ a.e. } r\geq0.
\end{equation}
That is, $\rho^r$ corresponds to the limit of the solution $\rho_t^r$ of \eqref{eq:pde_muerte} starting from an initial condition $(\rho_0^r)_{r\geq0}$ with finite mass.
\end{proposition}

\begin{proof}
Suppose that $\rho^r$ is a density with finite mass, solution of \eqref{eq:rhos_muerte_001}. Thus, $r\mapsto\rho^r$ is an integrable function, and $\rho^r$ converges to 0 in a set whose complement has a null Lebesgue measure.

Define $Q(r)=q((0,r])$, $r\geq0$. Evaluating $s=0$ and applying the Laplace transform with respect to the variable $r$ in \eqref{eq:rhos_muerte_001}, with $N^{(0,r]}=\int_{(0,r]}\rho^{\alpha}d\alpha$, we obtain
\begin{equation*}
\begin{split}
\mathcal L (\rho)(\xi) =&\, \frac{\rho^0}{\xi-d}-\lambda\mathcal L(Q)(\xi) - \lambda d \frac{\mathcal L(Q)(\xi)}{\xi-d}\\
=&\, \rho^0\mathcal L(e^{dr})(\xi)-\lambda\mathcal L (Q)(\xi)-\lambda d \mathcal L (Q)(\xi)\cdot \mathcal L(e^{dr})(\xi),
\end{split}
\end{equation*}
which implies, applying inverse Laplace transform, that
\begin{equation}\label{eq:rho_r_pre}
\rho^r = \rho^0e^{dr}-\lambda Q(r)-\lambda d Q\ast e^{dr}(r),
\end{equation}
where the operator $\ast$ stands for convolution. By definition,
\begin{equation}\label{eq_aux01}
(Q\ast e^{dr})(r) = \int_0^r Q(\alpha)e^{d(r-\alpha)}d\alpha = e^{dr}\int_0^r Q(\alpha)e^{-d\alpha}d\alpha.
\end{equation}

Note that $Q(\cdot)$ is a right-continuous function with left limits at every point of the interval $[0,\infty)$ with at most a finite number of discontinuity points (under Assumption \ref{ass:measures}). Thus, $Q(\cdot)$ defines a Lebesgue-Stieltjes measure $\mu_Q(\cdot)$, which coincides with $q(\cdot)$ on $\mathcal B([0,\infty))$. Thus, integrating by parts \cite[Proposition 5.3.3]{cohn2013}, for $f$ continuously differentiable,
\begin{equation}\label{eq:integracion_por_partes}
\int_{[0,r]}f(\alpha)\mu_Q(d\alpha) = Q(r)f(r)-Q(0)f(0)-\int_0^r Q(\alpha)f'(\alpha)d\alpha.
\end{equation}

Taking $f'(\alpha)=e^{-d\alpha}$ in \eqref{eq:integracion_por_partes} and noting that $Q(0)=q((0,0])=0$, we obtain
\begin{equation}\label{eq_aux02}
\int_0^rQ(\alpha)e^{-d\alpha}d\alpha = \frac{1}{d}\left( -Q(r)e^{-dr} + \int_{[0,r]}e^{-d\alpha}q(d\alpha) \right).
\end{equation}

Thus, combining \eqref{eq_aux01}, \eqref{eq_aux02} and \eqref{eq:rho_r_pre}, we obtain
\begin{equation}\label{eq:densidad_muerte_pre}
\rho^r \,=\, \rho^0e^{d r} - \lambda e^{d r}\int_{[0,r]}e^{-d\alpha}q(d\alpha).
\end{equation}

Dividing by $e^{d r}$ and taking limits as $r\nearrow\infty$ in \eqref{eq:densidad_muerte_pre}, we obtain
\begin{equation}\label{eq:rho0_muerte}
\rho^0 \,=\, \lambda\int_0^{\infty}e^{-d\alpha}q(d\alpha) \,=\, \lambda\hat q(d).
\end{equation}

Replacing \eqref{eq:rho0_muerte} in \eqref{eq:densidad_muerte_pre}, we obtain \eqref{eq:comp_densidades_rho_limit}.
\end{proof}

Let us define the relative entry and exit rates $\mu^{{\rm in}}$ and $\mu^{{\rm out}}$ as the entry and exit rates relative to the total prison population in equilibrium $N = \int_0^{\infty}\rho^{\alpha}d\alpha$:
\begin{equation}\label{eq:def_mus}
\mu^{{\rm in}}:=\frac{\lambda}{N}, \qquad \mu^{{\rm out}}:=\frac{\rho^0}{N},
\end{equation}
with units of time$^{-1}$. Recall that, under Assumption \ref{ass:measures}, the mean value of the probability measure $q(\cdot)$, denoted by $T_q$, is finite. We refer to $T_q$ as the mean initial sentence length.

In the remainder of this section, we obtain explicit formulas for the quantities of interest at steady state. Proposition \ref{prop:densidad_limite} allows us to characterize these quantities of interest in terms of the stationary density $\rho^r$ from the stationary equation \eqref{eq:rhos_muerte_001}.

\begin{proposition}\label{prop:N_mus_eq}
$\,$
\begin{enumerate}
\item The relative entry and exit rates $\mu^{{\rm in}}, \mu^{{\rm out}}$, defined in \eqref{eq:def_mus}, are linked by
\begin{equation}\label{eq:muin_muout_delta}
\mu^{{\rm out}} = \mu^{{\rm in}} - d.
\end{equation}
\item Suppose that $d>0$. The total number of prisoners in equilibrium is
\begin{equation}\label{eq:N_qhat}
N = \lambda\frac{1-\hat q(d)}{d}.
\end{equation}
\item Suppose that $d>0$. The relative rates $\mu^{{\rm in}}, \mu^{{\rm out}}$ have the explicit formulas
\begin{equation}\label{eq:muin_muout_death}
\mu^{{\rm in}} = \frac{d}{1-\hat q(d)} \quad,\quad \mu^{{\rm out}} = \frac{d\hat q(d)}{1-\hat q(d)} = \mu^{\rm in}\hat q(d).
\end{equation}
\end{enumerate}
\end{proposition}

\begin{proof}
In \eqref{eq:rhos_muerte_001}, evaluating $s=0$ and taking limits as $r\nearrow\infty$, since $\lim_{r\rightarrow\infty}\rho^{r}=0$, we obtain
\begin{equation}\label{eq:rho0_muerte_N}
\rho^{0} = \lambda q((0,\infty))-d N^{(0,\infty)} = \lambda - d N.
\end{equation}

Dividing \eqref{eq:rho0_muerte_N} by $N$, for the relative rates $\mu^{{\rm in}},\mu^{{\rm out}}$ defined in \eqref{eq:def_mus}, we obtain \eqref{eq:muin_muout_delta}. Now, combining \eqref{eq:rho0_muerte_N} and \eqref{eq:rho0_muerte}, we obtain \eqref{eq:N_qhat}. Replacing \eqref{eq:N_qhat} and \eqref{eq:rho0_muerte} in the definitions of $\mu^{{\rm in}}$ and $\mu^{{\rm out}}$ given in \eqref{eq:def_mus}, we obtain \eqref{eq:muin_muout_death}.
\end{proof}

\begin{remark}
The total population at equilibrium given by \eqref{eq:N_qhat} is a key value for the design of a penitentiary system, which involves all the data of the problem: the entry rate, removal rate from the prison and the initial sentence length distribution.
\end{remark}


Let us define $R_{\rho}$ as the mean remaining sentence time inside the prison. This term can be computed from $\rho^r$ as
\begin{equation*}
R_{\rho} = \frac{1}{N}\int_0^{\infty}r\rho^r dr.
\end{equation*}

\begin{proposition}
If $d>0$, the mean remaining sentence time inside the prison $R_{\rho}$ has the expression
\begin{equation}\label{eq:R_rho}
R_{\rho} = \frac{T_q}{1-\hat q(d)}-\frac{1}{d}.
\end{equation}
\end{proposition}

\begin{proof}
Swapping $r$ and $s$ and taking limits as $s\nearrow\infty$ in \eqref{eq:rhos_muerte_001}, we obtain
\begin{equation}\label{eq:aux001}
\rho^r = \lambda q((r,\infty))-d N^{(r,\infty)}.
\end{equation}

If we integrate \eqref{eq:aux001},
\begin{equation}\label{eq:aux002}
\int_{0}^{\infty}\rho^rdr = \lambda \int_0^{\infty}q((r,\infty))dr-d \int_0^{\infty}N^{(r,\infty)}dr.
\end{equation}

Note that $\tilde\rho^{r}=\rho^r/N$ is a probability density function on the set $\mathcal C$ corresponding to the density of the remaining time in prison. Thus, following Remark \ref{rem:Tq}, the mean sentence time inside prison $R_{\rho}$ can be computed as
\begin{equation}\label{eq:R_rho_explained}
R_{\rho}=\int_0^{\infty} \left( \int_{(r,\infty)}\tilde \rho^{\alpha}d\alpha \right) dr = \frac{1}{N}\int_0^{\infty} \left( \int_{(r,\infty)}\rho^{\alpha}d\alpha \right) dr = \frac{1}{N}\int_0^{\infty} N^{(r,\infty)} dr.
\end{equation}

Thus, from \eqref{eq:Tq} and \eqref{eq:R_rho_explained}, \eqref{eq:aux002} translates into
\begin{equation}\label{eq:relacion_N_tiempos_muerte}
N = \lambda T_q - d N R_{\rho}.
\end{equation}

Since $d>0$, replacing $N$ from \eqref{eq:N_qhat} and isolating $R_{\rho}$ in \eqref{eq:relacion_N_tiempos_muerte}, we obtain \eqref{eq:R_rho}.
\end{proof}


To obtain the values deduced in Proposition \ref{prop:N_mus_eq} when there is no removal from the prison (i.e., $d=0$), let us define the following functions:

\begin{equation}\label{eq:functions_eq}
\begin{split}
d\mapsto N=\left\{\begin{array}{lr}
\lambda\frac{1-\hat q(d)}{d},& d>0,\\[2mm]
\lambda T_q,& d=0,
\end{array}\right.
\quad
d\mapsto \mu^{{\rm in}}=\left\{\begin{array}{lr}
\frac{d}{1-\hat q(d)},& d>0,\\[2mm]
\frac{1}{T_q},& d=0,
\end{array}\right.\\[3mm]
d\mapsto \mu^{{\rm out}}=\left\{\begin{array}{lr}
\frac{d\hat q(d)}{1-\hat q(d)},& d>0,\\[2mm]
\frac{1}{T_q},& d=0.
\end{array}\right.
\end{split}
\end{equation}

\begin{proposition}\label{rem:N_tasas_nomuerte}
The functions $d\mapsto N$, $d\mapsto \mu^{{\rm in}}$, and $d\mapsto \mu^{{\rm out}}$ defined in \eqref{eq:functions_eq} are right-continuous at $d=0$.
\end{proposition}

\begin{proof}
When taking $d=0$, from \eqref{eq:relacion_N_tiempos_muerte}, we directly obtain $N=\lambda T_q$. From this equation and \eqref{eq:rho0_muerte_N}, we directly obtain $\rho^0=\lambda=N/T_q$, and dividing by $N$, we obtain $\mu^{{\rm in}} = \mu^{{\rm out}} = 1/T_q$. 

To prove the continuity of $N$, $\mu^{{\rm in}}$ and $\mu^{{\rm out}}$ at $d=0$, we note that
$\hat q(0) = \int_0^{\infty} q(dr) = q((0,\infty)) = 1$. 
%
Thus,
\begin{equation*}
\lim_{d\searrow0} \frac{1-\hat q(d)}{d} = -\lim_{d\searrow0} \frac{\hat q(d)-\hat q(0)}{d-0} = -\hat q' (0),
\end{equation*}
where
\begin{equation*}
\hat q'(d) =  \int_{0}^{\infty}\frac{\partial}{\partial d}e^{-d r}q(dr) = \int_{0}^{\infty}-re^{-d r}q(dr),
\end{equation*}
which implies
$\hat q ' (0) = \int_{0}^{\infty}-rq(dr) = -T_q$.
%
Then,
\begin{equation*}\label{eq:limite_division_qhat}
\lim_{d\searrow0} \frac{1-\hat q(d)}{d} = -\hat q ' (0) = T_q,
\end{equation*}

Note that, taking limits as $d\searrow0$ in $N,\mu^{{\rm in}},\mu^{{\rm out}}$ in \eqref{eq:functions_eq},
\begin{equation*}
N = \lambda\frac{1-\hat q(d)}{d} \rightarrow \lambda T_q,\quad \mu^{{\rm in}}=\frac{d}{1-\hat q(d)} \rightarrow \frac{1}{T_q},\quad \mu^{{\rm out}} = \mu^{\rm in}\hat q(d) \rightarrow \frac{1}{T_q},
\end{equation*}
which proves the result.
\end{proof}

\begin{remark}\label{remark:caso_limite_d_0}
As a consequence of Proposition \ref{rem:N_tasas_nomuerte}, we henceforth develop the results using the obtained formulas for the case of $d>0$, taking into consideration that the case of $d=0$ is obtained as the limiting result when $d$ converges to 0.
\end{remark}

\begin{remark}
Typically, in mathematical models of communicable diseases for penitentiary systems, the sentence length structure of inmates is ignored \cite{Beauparlant:2016,Brauer:2001,castillo2020,Legrand:2008,Witbooi:2017}, and the entry (and exit) rate of new inmates is assumed to be the inverse of the average sentencing length. In other words, the relative rate of entry $\mu^{{\rm in}}$, with respect to the total population, is taken as the inverse of the mean time of residence, which in this case corresponds to $T_q$. In Propositions \ref{prop:N_mus_eq} and \ref{rem:N_tasas_nomuerte}, we show that this is true when the removal rate is null. If the removal rate is not null, $\mu^{{\rm in}}$ is no longer a function of the mean time of residence| $T_q$ but depends on the complete distribution of sentence length $q(\cdot)$ via its Laplace transform. A high removal rate can be interpreted as a bad situation, for instance, if this rate consists mainly of deaths and transfers due to problems in the prison (overcrowding, riots), therefore, to assume that this rate is close to zero is almost to assume an ideal situation.
\end{remark}

\subsection{Population models with two classes of sentence lengths}\label{subsec:clases}

Suppose that we want to classify the inmates according to the length of their initial sentence while not affecting the homogeneous mixing among all prisoners (which is important from the epidemiological perspective to be developed later). This means the existence of a threshold $\rstar>0$, such that every sentence in the set $\C_1 := (0,\rstar]$ (resp. $\C_2 := (\rstar,\infty)$) is considered a short (resp. long) sentence length. Note that $\C = \C_1\cup C_2$ is a disjoint union. We say that a prisoner belongs to the class $i$, or his/her type is $i$, if his/her initial sentence length belongs to $\C_i$ ($i=1,2$).

To apply the results of the previous section, we note that each class has its own entry rate and its own distribution of initial sentence length. Indeed, let us define
\begin{equation}\label{eq:p_i}
p_t^i\,:=\,q_t(\C_i),\qquad i=1,2,
\end{equation}
the proportions of individuals entering the prison with a short ($i=1$) or long ($i=2$) initial sentence length. Thus, the entry rate of prisoners to class $i$ is $\lambda_t^i:=\lambda_t q_t(\C_i)=\lambda_t p_t^i$, and the distribution of initial sentence lengths of class $i$ corresponds to the distribution of initial sentence lengths, conditional on the inmate entering said class:
\begin{equation}\label{eq:q_i}
q_t^i(I) \,:=\, \frac{q_t(I\cap \C_i)}{q_t(\C_i)} \, = \,\frac{q_t(I\cap \C_i)}{p_t^i} ,\quad I\in\mathcal B(\C).
\end{equation}

Suppose that the removal rate for each class is $d^i_t\geq0$ for all $t\geq0$. In this setting, from \eqref{eq:balance_poblacion}, the total number of prisoners in each class $N^1_t,N^2_t$ follows the equations
\begin{equation}\label{eq:balance_poblacioN^i}
\frac{d}{dt} N_t^1 = \lambda_{t}^1 - \rho_t^{0,1} - d_t^1 N_t^1,\quad \frac{d}{dt} N_t^2 = \lambda_{t}^2 - \rho_t^{0,2} - d_t^2 N_t^2, \quad\mbox{ a.e. }t\geq0,
\end{equation}
where $\rho_t^{0,1}$ and $\rho_t^{0,2}$ are the exit rates of classes 1 and 2, respectively, and solve the McKendrick equation \eqref{eq:pde_muerte}, that is,
\begin{equation}\label{eq:pde_muerte_i}
\frac{\partial \rho_t^{r,1}}{\partial t} + d_t^1\rho_t^{r,1} \,  = \frac{\partial\rho_t^{r,1}}{\partial r} + \lambda_{t}^1 q_t^{1\prime}(r),\quad 
\frac{\partial \rho_t^{r,2}}{\partial t} + d_t^2\rho_t^{r,2} \,  = \frac{\partial\rho_t^{r,2}}{\partial r} + \lambda_{t}^2 q_t^{2\prime}(r), 
\end{equation}
where $q_t^{i\prime}(\cdot)$ denotes the density of the distribution $q_t^{i}(\cdot)$ at time $t$, $i=1,2$.

\begin{lemma}\label{lemma:N_Ni}
Suppose that $d^1_t=d^2_t=:d_t\geq0$ for all $t\geq0$. Then, the density of the total population $\rho_t^r$ satisfies
\begin{equation}\label{eq:pde_muerte_sum}
\frac{\partial \rho_t^{r}}{\partial t} + d_t\rho_t^{r} \,  = \frac{\partial\rho_t^{r}}{\partial r} + \lambda_{t} q'_{t}(r).
\end{equation}
Consequently, the total population $N_t$ inside the prison satisfies
\begin{equation}\label{eq:balance_poblacion_sum}
\frac{d}{dt} N_{t} = \lambda_{t} - \rho_t^{0} - d_t N_t.
\end{equation}
\end{lemma}

\begin{proof}
Adding the two equations in \eqref{eq:pde_muerte_i}, we obtain
\begin{equation*}
\frac{\partial }{\partial t}(\rho_t^{r,1}+\rho_t^{r,2}) + d_t(\rho_t^{r,1}+\rho_t^{r,2}) \,  = \frac{\partial}{\partial r}(\rho_t^{r,1}+\rho_t^{r,2}) + \lambda_{t} (p_t^1q_t^{1\prime}(r)+p_t^2 q_t^{2\prime}(r)).
\end{equation*}

Note that, from \eqref{eq:p_i} and \eqref{eq:q_i},
\begin{equation}\label{eq:rel_qi_q}
p_t^1q_t^{1\prime}(r)+p_t^2 q_t^{2\prime}(r) = p_t^1\frac{q_t'(r)}{p_t^1}\indi{\mathcal C_1}+p_t^2\frac{q_t'(r)}{p_t^2}\indi{\mathcal C_2} = q_t'(r).
\end{equation}
Since $\rho_t^r=\rho_t^{r,1}+\rho_t^{r,2}$ is the density of the total prison population, we conclude \eqref{eq:pde_muerte_sum}.\\

Now, adding the two equations in \eqref{eq:balance_poblacioN^i}, we obtain
\begin{equation*}
\frac{d}{dt} (N_t^1+N_t^2) = \lambda_{t}(p_t^1+p_t^2) - (\rho_t^{0,1}+\rho_t^{0,2}) - d_t (N_t^1+N_t^2).
\end{equation*}
Since $N=N^1+N^2$ is the total prison population and in the previous computations we obtained that $\rho^{0}=\rho^{0,1}+\rho^{0,2}$, we conclude \eqref{eq:balance_poblacion_sum}.
\end{proof}

Let us define $\pi_t^i:=N_t^i/N_t$ as the proportions of inmates of each class $i=1,2$, relative to the total prison population at each time $t$. Define also the entry/exit rates of each class and the entry/exit rates of the whole prison relative to the class size:
\begin{equation}\label{eq:tasas_t}
\mu_t^{{\rm in}}:= \frac{\lambda_t}{N_t},\quad
\mu_t^{{\rm out}}:= \frac{\rho^0_t}{N_t},\quad
\mu_t^{{\rm in},i} := \frac{\lambda^i_t}{N_t^i},\quad \mu_t^{{\rm out},i} := \frac{\rho^{0,i}_t}{N_t^i},\quad i=1,2. 
\end{equation}

\begin{remark}
Suppose that $d_t^1=d_t^2$ for all $t\geq0$. Then, we have the following relations between the entry/exit rates of each class and the entry/exit rates of the whole prison:
\begin{equation}\label{eq:aux005}
\left\{\quad\begin{split}
\mu_t^{{\rm in},1}\pi_t^1 + \mu_t^{{\rm in},2}\pi_t^2 \,=&\, \frac{\lambda_t^1}{N_t^1}\frac{N_t^1}{N_t} + \frac{\lambda_t^2}{N_t^2}\frac{N_t^2}{N_t} \,=\, \frac{\lambda_t^1+\lambda_t^2}{N_t} \,=\, \frac{\lambda_t}{N_t} \,=\, \mu_t^{{\rm in}},\\[2mm]
\mu_t^{{\rm out},1}\pi_t^1 + \mu_t^{{\rm out},2}\pi_t^2 \,=&\, \frac{\rho_t^{0,1}}{N_t^1}\frac{N_t^1}{N_t} + \frac{\rho_t^{0,2}}{N_t^2}\frac{N_t^2}{N_t} \,=\, \frac{\rho_t^{0,1}+\rho_t^{0,2}}{N_t} \,=\, \frac{\rho_t^{0}}{N_t} \,=\, \mu_t^{{\rm out}}.
\end{split}\right.
\end{equation}
\end{remark}

\subsubsection{Steady-state analysis for the two-class model}
We are now interested in obtaining explicit expressions and relations between the sizes of the different classes in equilibrium. Suppose that the entry rate $\lambda$, the distribution $q(\cdot)$, and the removal rate $d$ do not depend on $t$ and that the system operates at steady state. In this setting, \eqref{eq:N_qhat} applies to each class separately. From \eqref{eq:N_qhat}, the total number of prisoners in class $i$ is
\begin{equation}\label{eq:N^i_muerte}
N^i = \frac{\lambda^i}{d^i}(1-\hat{q}^i(d^i)),\quad i=1,2,
\end{equation}
where $\hat{q}^i(\cdot)$ denotes the Laplace transform of the distribution of initial sentence lengths with respect to each class $q^i(\cdot)$. If we suppose that the removal rates are equal for the two classes, namely, $d^1=d^2=:d>0$, from Lemma \ref{lemma:N_Ni}, we have $N = N^1 + N^2$, where, according to \eqref{eq:N_qhat},
\begin{equation}\label{eq:N_muerte}
N = \frac{\lambda}{d}(1-\hat{q}(d)).
\end{equation}

Then, the proportion of prisoners of each class in equilibrium is $\pi^i = N^i/N$, where, from \eqref{eq:N^i_muerte} and \eqref{eq:N_muerte}, we have
\begin{equation}\label{eq:pi_i_muerte}
\pi^i = p^i\frac{1-\hat{q}^i(d)}{1-\hat q(d)},\quad i=1,2,
\end{equation}
where $p^i=q(\C_i)$, as in \eqref{eq:p_i}. Additionally, the entry rates $\muini= \lambda^i/N^i$, $\mu^{{\rm in}}=\lambda/N$, relative to the stationary population sizes $N^i$, $N$, (as defined in \eqref{eq:tasas_t}) in equilibrium, are (from \eqref{eq:muin_muout_death})
\begin{equation}\label{eq:muin_comp_death}
\mu^{{\rm in}} = \frac{d}{1-\hat q(d)} ,\quad \mu^{{\rm in},1} = \frac{d}{1-\hat{q}^1(d)},\quad \mu^{{\rm in},2} = \frac{d}{1-\hat{q}^2(d)}.
\end{equation}

In a similar way, the exit rates $\muouti=\rho^{0,i}/N^i$, $\mu^{{\rm out}}=\rho^{0}/N$, relative to the stationary population sizes $N^i$, $N$ in equilibrium, are (from \eqref{eq:muin_muout_death})
\begin{equation}\label{eq:muout_i_death}
\mu^{{\rm out}} = \frac{d\hat q(d)}{1-\hat q(d)}, \quad \mu^{{\rm out},1} = \frac{d\hat{q}^1(d)}{1-\hat{q}^1(d)}, \quad \mu^{{\rm out},2} = \frac{d\hat{q}^2(d)}{1-\hat{q}^2(d)}.
\end{equation}

\begin{lemma}
Suppose that $d_1=d_2=d>0$. Then,
\begin{equation*}\label{eq:relacion_mus_delta}
\frac{1}{\mu^{{\rm in}}} = \frac{p^1}{\mu^{{\rm in},1}} + \frac{p^2}{\mu^{{\rm in},2}}.
\end{equation*}
\end{lemma}

\begin{proof}
From \eqref{eq:muin_comp_death}, following a similar calculation as in \eqref{eq:rel_qi_q},
\begin{equation*}
\begin{split}
\frac{p^1}{\mu^{{\rm in},1}} + \frac{p^2}{\mu^{{\rm in},2}} = &\, p^1\frac{1-\hat{q}^1(d)}{d} + p^2\frac{1-\hat{q}^2(d)}{d} \\
=&\, \frac{1}{d}\left((p^1+p^2)- (p^1\hat{q}^1(d) + p^2\hat{q}^2(d)\right) = \frac{1-\hat q(d)}{d} = \frac{1}{\mu^{{\rm in}}}.
\end{split}
\end{equation*}$\,$
\end{proof}

\begin{corollary}\label{remark:orden_muis}
We have $\pi^1\leq p^1$, $\pi^2\geq p^2$, and $\mu^{{\rm in},2}\leq\mu^{{\rm in}}\leq\mu^{{\rm in},1}$.
\end{corollary}

\begin{proof}
These results are obtained from the relation $\hat{q}^2(d)\leq\hat q(d)\leq\hat{q}^1(d)$. Indeed,
\begin{equation*}
\begin{split}
\hat q(d)-\hat{q}^2(d) =&\,  \int_0^{\infty}e^{-d r}q(dr) - \frac{1}{p^2}\int_{r^{\star}}^{\infty}e^{-d r}q(dr)\\
=&\,\left(1-\frac{1}{p^2}\right)\int_{r^{\star}}^{\infty}e^{-d r}q(dr) + \int_0^{r^{\star}}e^{-d r}q(dr) \\
\geq&\,\left(1-\frac{1}{p^2}\right)e^{-d r^{\star}} p^2 + e^{-d r^{\star}}p^1 \\
\geq&\,e^{-d r^{\star}} (p^2-1+p^1) = 0.
\end{split}
\end{equation*}
The proof of the other inequality is analogous. Thus, the results follow.
\end{proof}

\begin{remark}\label{rem:Ni_tasasi_nomuerte}
From Proposition \ref{rem:N_tasas_nomuerte}, taking limits as $d\searrow0$ in \eqref{eq:N^i_muerte}, \eqref{eq:N_muerte}, the proportions of prisoners that belong to each class, relative to the total number of prisoners, converge to
\begin{equation*}\label{eq:pi_i}
\pi^i := \frac{N^i}{N} \rightarrow \frac{\lambda^i T_{q^i}}{\lambda T_q} = p^i\frac{T_{q^i}}{T_q}, \quad i=1,2,
\end{equation*}
with $T_{q^i}$ being the mean initial sentence length relative to class $i$ (that is, the mean of $q^i(\cdot)$), which corresponds to the case without removal (see Proposition \ref{rem:N_tasas_nomuerte} and Remark \ref{remark:caso_limite_d_0}). Additionally, the entry and exit rates relative to the stationary population sizes from \eqref{eq:muin_comp_death}, \eqref{eq:muout_i_death}, converge to
\begin{equation*}\label{eq:mus_y_Ts}
\mu^{{\rm in}},\mu^{{\rm out}} \rightarrow \frac{1}{T_q}, \quad \mu^{{\rm in},1},\mu^{{\rm out},1} \rightarrow \frac{1}{T_{q^1}}, \quad \mu^{{\rm in},2},\mu^{{\rm out},2} \rightarrow \frac{1}{T_{q^2}}.
\end{equation*}
\end{remark}

\begin{remark}\label{rem:more_classes}
The previous procedure can be performed for a finite arbitrary number of classes, obtaining analogous results. Additionally, the classification in a prison population can be performed using other criteria. For instance, one can classify individuals into dangerous and nondangerous inmates or, related to a communicable disease, a possible classification can be a high-risk population (superspreader individuals) and a low-risk population.
\end{remark}

\subsection{Examples}\label{subsec:examples}

\subsubsection{Discrete initial sentence length distribution}

Consider the constant entry rate $\lambda_t=\lambda>0$ and death rate $d_t=d\geq0$, and suppose that the distribution $q(\cdot)$ is discrete, with only two possible initial sentence lengths $0<T_1<T_2$, with probabilities $p^1$ and $p^2$, respectively, such that $p^1+p^2=1$, that is,
\begin{equation*}
q'(r)=p^1 \delta_{T_1}'(r) + p^2 \delta_{T_2}'(r), \quad r\geq0,
\end{equation*}
with $\delta_{x}'(\cdot)$, the Dirac function concentrated on $x$ associated with the Dirac measure $\delta_{x}(\cdot)$ concentrated on $\{x\}$. Then, $T_q = p^1T_1+p^2T_2$, $\hat q(d) = p^1e^{-d T_1} + p^2e^{-d T_2}$. The total population, entry and exit rates at steady state are shown in Table \ref{table:discrete}.

\begin{table}[ht]
\vspace{.2cm}
\begin{center}
\caption{Total population, entry and exit rates at steady state for a discrete initial sentence length distribution.}\label{table:discrete}
\begin{tabular}{|c|c|c|}
\hlineB{3}
{\bf Quantity} & {\bf $d>0$} & $d=0$ \\ 
\hlineB{3}
&&\\[-2mm]
$N$ & $\dfrac{\lambda}{d}(1-(p^1e^{-d T_1} + p^2e^{-d T_2}))$ 	& 	$\lambda(p^1T_1+p^2T_2)$ \\ [2mm]
$\mu^{\rm in}$ & $\dfrac{d}{1-(p^1e^{-d T_1} + p^2e^{-d T_2})}$ & $\dfrac{1}{p^1T_1+p^2T_2}$\\[6mm]
$\mu^{\rm out}$ & $\dfrac{d(p^1e^{-d T_1} + p^2e^{-d T_2})}{1-(p^1e^{-d T_1} + p^2e^{-d T_2})}$ & $\dfrac{1}{p^1T_1+p^2T_2}$\\[-2mm]
&&\\
\hlineB{3}
\end{tabular}
\end{center} 
\end{table}

In this example, the population density with respect to the remaining sentence time is given by $\rho^0 = \lambda\hat q(d) = \lambda(p^1e^{-d T_1} + p^2e^{-d T_2})$ and
\begin{equation*}
\rho^r = \lambda \int_r^{\infty}e^{-d(\alpha-r)}q'(\alpha)d\alpha=\left\{\begin{split}
\quad\lambda (p^1e^{-d (T_1-r)}+p^2e^{-d (T_2-r)}),&\quad r<T_1,  \\
\lambda p^2e^{-d (T_2-r)},&\quad T_1\leq r<T_2, \\
0&\quad T_2\leq r, 
\end{split}\right.\hfill
\end{equation*}
for both cases $d>0$ and $d=0$.

If $\C_1 = (0, r^{\star}]$, $\C_2 = (r^{\star},\infty)$, with $T_1 < r^{\star}< T_2$, then $T_{q^1}=T_1$, $T_{q^2}=T_2$, and we obtain the expressions in Table \ref{table:discrete-classes}.
\begin{table}[ht]
\vspace{.2cm}
\begin{center}
\caption{Total population, entry and exit rates at steady state for a discrete initial sentence length distribution considering two classes.}\label{table:discrete-classes}
\begin{tabular}{|c|c|c|c|c|}
\hlineB{3}
\multirow{2}{*}{\bf Quantity} & \multicolumn{2}{|c|}{$d>0$} & \multicolumn{2}{|c|}{$d=0$} \\ \clineB{2-5}{3}
&&&&\\[-4mm]
& {\bf Class 1} & {\bf Class 2} & {\bf Class 1} & {\bf Class 2} \\ \clineB{2-5}{3}
\hlineB{3}
&&\\[-2mm]
$\hat{q}^i(d)$ & $e^{-d T_1}$ & $e^{-d T_2}$  & $1$ & $1$\\[2mm]
$N^i$ & $\dfrac{\lambda p^1}{d}(1-e^{-d T_1})$ & $\dfrac{\lambda p^2}{d}(1-e^{-d T_2})$ & $\lambda p^1 T_1$  & $\lambda p^2 T_2$ \\ [2mm]
$\mu^{{\rm in},i}$ & $\dfrac{d}{1-e^{-d T_1}}$ & $\dfrac{d}{1-e^{-d T_2}}$ & $\dfrac{1}{T_1}$ & $\dfrac{1}{T_2}$ \\ [3mm]
$\mu^{{\rm out},i}$ & $\dfrac{d e^{-d T_1}}{1-e^{-d T_1}}$ & $\dfrac{d e^{-d T_2}}{1-e^{-d T_2}}$ & $\dfrac{1}{T_1}$ & $\dfrac{1}{T_2}$\\[-2mm]
&&\\
\hlineB{3}
\end{tabular}
\end{center} 
\end{table}

\subsubsection{Continuous initial sentence length distribution}

Consider a constant entry rate $\lambda_t=\lambda>0$ and a death rate $d_t=d\geq0$, and suppose that $q(\cdot)$ follows an exponential distribution with rate $c>0$, that is, its probability density function corresponds to
\begin{equation*}
q'(r)=ce^{-cr}, \quad r\geq0.
\end{equation*}
Then, $T_q = \dfrac{1}{c}$ and $\hat q(d) = \dfrac{c}{d + c}$, obtaining the expressions in Table \ref{table:exponential}. The population density function with respect to the remaining sentence time is
\begin{equation*}
\rho^r = \dfrac{\lambda c}{d + c}e^{-cr} ,\quad r\geq0. 
\end{equation*}

\begin{table}[ht]
\vspace{.2cm}
\begin{center}
\caption{Total population, entry and exit rates at steady state for an exponential sentence length distribution.}\label{table:exponential}
\begin{tabular}{|c|c|c|}
\hlineB{3}
{\bf Quantity} & $d>0$ & $d=0$ \\ 
\hlineB{3}
&&\\[-2mm]
$N$ & $\dfrac{\lambda}{d+c}$ 	& 	$\dfrac{\lambda}{c}$ \\ [3mm]
$\mu^{\rm in}$ & $d + c$ & $c$\\[2mm]
$\mu^{\rm out}$ & $c$ & $c$\\[-2mm]
&&\\
\hlineB{3}
\end{tabular}
\end{center} 
\end{table}

If $\mathcal C_1 = (0,r^{\star}]$, $\mathcal C_2 = (r^{\star},\infty)$, then $p^1 = 1-e^{-cr^{\star}}$, $p^2 = e^{-cr^{\star}}$, $T_{q^1}= \frac{1}{c}-\frac{r^{\star}e^{-cr^{\star}}}{1-e^{-cr^{\star}}}$, and $T_{q^2}=\frac{1}{c}+r^{\star}$. In Table \ref{table:exponential-classes1}, we report the corresponding expressions when $d>0$ and we do so in Table \ref{table:exponential-classes2} when $d=0$.

\begin{table}[ht]
\vspace{.2cm}
\begin{center}
\caption{Total population, entry and exit rates at steady state for an exponential sentence length distribution considering two classes when $d > 0$.}\label{table:exponential-classes1}
\resizebox{\textwidth}{!}{%
\begin{tabular}{|c|c|c|}
\hlineB{3}
{\bf Quantity} & {\bf Class 1} & {\bf Class 2}\\ 
\hlineB{3}
&&\\[-2mm]
$\hat{q}^i(d)$ & $\dfrac{c}{d+c}\dfrac{1-e^{-(d + c)r^{\star}}}{1-e^{-cr^{\star}}}$	&	$\dfrac{c}{d+c}\dfrac{e^{-(d + c)r^{\star}}}{e^{-cr^{\star}}}$\\[3mm]
$N^i$ &  $\dfrac{\lambda}{d}\left( (1-e^{-cr^{\star}}) - \dfrac{c}{d+c}(1-e^{-(d+c)r^{\star}}) \right)$
&	$\dfrac{\lambda}{d}\left( e^{-cr^{\star}} - \dfrac{c}{d+c}e^{-(d+c)r^{\star}} \right)$ \\[3mm]
$\mu^{{\rm in},i}$ &  $d\left(  1 - \frac{c}{d+c}\frac{1-e^{-(d+c)r^{\star}}}{1-e^{-cr^{\star}}}  \right)^{-1}$ 	&	$d\left( 1 - \frac{c}{d+c}e^{-d r^{\star}} \right)^{-1}$\\[4mm]
$\mu^{{\rm out},i}$ &  $\dfrac{cd}{d+c}\dfrac{1-e^{-(d + c)r^{\star}}}{1-e^{-cr^{\star}}} \left(  1 - \frac{c}{d+c}\frac{1-e^{-(d+c)r^{\star}}}{1-e^{-cr^{\star}}} \right)^{-1}$ 	& 	$\dfrac{cd}{d+c}\dfrac{e^{-(d + c)r^{\star}}}{e^{-cr^{\star}}} \left( 1 - \frac{c}{d+c}e^{-d r^{\star}} \right)^{-1}$ \\[-2mm]
&&\\
\hlineB{3}
\end{tabular}}
\end{center} 
\end{table}

\begin{table}[ht]
\vspace{.2cm}
\begin{center}
\caption{Total population, entry and exit rates at steady state for an exponential sentence length distribution considering two classes when $d=0$.}\label{table:exponential-classes2}
\begin{tabular}{|c|c|c|}
\hlineB{3}
{\bf Quantity} & {\bf Class 1} & {\bf Class 2}\\ 
\hlineB{3}
&&\\[-4mm]
$N^i$  &  $\lambda\left( \dfrac{1-e^{-cr^{\star}}}{c} - r^{\star}e^{-cr^{\star}} \right)$ &	$\lambda e^{-cr^{\star}}\left( r^{\star} + \dfrac{1}{c} \right)$ \\[3mm]
$\mu^{{\rm in},i}$  &  $\left(\dfrac{1}{c}-\dfrac{r^{\star}e^{-cr^{\star}}}{1-e^{-cr^{\star}}}\right)^{-1}$ 	&	$\left(\dfrac{1}{c}+r^{\star}\right)^{-1}$\\[4mm]
$\mu^{{\rm out},i} $   &   $\left(\dfrac{1}{c}-\dfrac{r^{\star}e^{-cr^{\star}}}{1-e^{-cr^{\star}}}\right)^{-1}$ 	& 	$\left(\dfrac{1}{c}+r^{\star}\right)^{-1}$ \\[-2mm]
&&\\
\hlineB{3}
\end{tabular}
\end{center} 
\end{table}

\subsection{Example of application: estimation of initial sentence distribution}\label{subsec:examples}

It is often difficult to have access to the distribution $q(\cdot)$ of initial sentence lengths. Nevertheless, it is possible that the information of the operation of the prison is stored or reported in the form of periodic snapshots or averages of its status, consisting of the entry rate (new inmates or transferred), exit rate (by finishing the sentence), removal rate (by transfer, death, pardon or commutation of the sentence of prisoners), and histograms of the initial sentence lengths relative to the existing population. 

Suppose that our source of information about the prison contains the aggregated entry rate $\lambda_{\rm data}$, the removal rate $d_{\rm data}$, the exit rate $\rho_{\rm data}^0$, and a histogram of the current state of the prison, consisting of the frequencies $N_{\rm data}^i$ of inmates with respect to their initial sentence length, split in $n$ classes corresponding to the intervals $\mathcal C_i=(T_{i-1},T_i]$, with $0=T_0<\dots<T_n=T_{\rm max}<\infty$. We wish to estimate the initial sentence distribution $q(\cdot)$ from the known data. For this, we consider the same interval classification of sentence lengths, and write $q(\cdot)=\sum_{i=1}^n p^iq^i(\cdot)$, where the probabilities by class $(p^i)_{i=1}^n$ and the distributions conditional to the classes $(q^i(\cdot))_{i=1}^n$ are unknown.

Notice that the information given in the histogram corresponds to the number of inmates whose initial sentence length belongs to the interval $\mathcal C_i$. Thus, defining $N_{\rm data}:=\sum_{i=1}^n N_{\rm data}^i$, the proportion $\pi_{\rm data}^i:=N_{\rm data}^i/N_{\rm data}$ is an estimator of the proportion $\pi^i$ of prisoners at each class.

Based on \eqref{eq:pi_i_muerte} and Remark \ref{rem:more_classes}, we can obtain an estimator of $(p^i)_{i=1}^n$. Indeed, if $(\pi^i)_{i=1}^n$, $(q(\cdot)^i)_{i=1}^n$, and $d$ were known,
from \eqref{eq:pi_i_muerte} we would obtain
\begin{equation*}
1-\hat q(d) =  p^i\frac{1-\hat{q}^i(d)}{\pi^i}=p^j\frac{1-\hat{q}^j(d)}{\pi^j},\quad \forall i,j=1,\dots,n,
\end{equation*}
which implies, choosing a particular (fixed) index $j^{\star}\in\{1,\dots,n\}$, that
\begin{equation*}
p^i = p^{j^{\star}}\frac{\pi^i}{\pi^{j^{\star}}}\frac{1-\hat{q}^{j^{\star}}(d)}{1-\hat{q}^i(d)},\quad \forall i=1,\dots,n.
\end{equation*}
Imposing $\sum_{i=1}^np^i=1$, we obtain $p^{j^{\star}}=\frac{\pi^{j^{\star}}}{1-\hat q^{j^{\star}}(d)}\left(\sum_{j=1}^n \frac{\pi^j}{1-\hat q^j(d)}\right)^{-1}$, and then,
\begin{equation}
p^i = \frac{\pi^i/(1-\hat q^i(d))}{\sum_{j=1}^n \pi^j/(1-\hat q^j(d))},\quad i=1,\dots,n.
\end{equation}

Given the discrete nature of a histogram, we can suppose that for each $i\in\{1,\dots,n\}$ the distribution $q^i(\cdot)$ is concentrated in a point $S_i\in(T_{i-1},T_i]$ which operates as a representative of the interval $\mathcal C_i=(T_{i-1},T_i]$, that is, $q^i(\cdot)=\delta_{S_i}(\cdot)$, whose Laplace transform is $\hat q^i(d)=e^{-dS_i}$. Then, our method to estimate $(p^i)_{i=1}^n$ is:
\begin{enumerate}
\item Compute $\pi_{\rm data}^i:=N_{\rm data}^i/N_{\rm data}$
\item Compute $\displaystyle \tilde p^i = \frac{\pi_{\rm data}^i/(1-e^{-d_{\rm data}S_i})}{\sum_{j=1}^n \pi_{\rm data}^j/(1-e^{-d_{\rm data}S_j})}= \frac{N_{\rm data}^i/(1-e^{-d_{\rm data}S_i})}{\sum_{j=1}^n N_{\rm data}^j/(1-e^{-d_{\rm data}S_j})}$.
\end{enumerate}

It is possible to compare the theoretical quantities (given by the model) with those computed from the real data. For instance, using $\hat q(d) = \sum_{i=1}^n\tilde p^ie^{-dS_i}$ the Laplace transform of the estimated distribution $q(\cdot)$, we can:
\begin{enumerate}
\item compare $N_{\rm data}$ with $\tilde N_{\rm model} = \lambda_{\rm data}\frac{1-\hat{q}(d_{\rm data})}{d_{\rm data}}$.
\item compare $\mu_{\rm data}^{\rm in}=\lambda_{\rm data}/N_{\rm data}$ with $\tilde \mu_{\rm model}^{\rm in} = \frac{d_{\rm data}}{1-\hat q(d_{\rm data})}$.
\item compare $\mu_{\rm data}^{\rm out}=\rho_{\rm data}^0/N_{\rm data}$ with $\tilde \mu_{\rm model}^{\rm out} = \frac{d_{\rm data}\hat q(d_{\rm data})}{1-\hat q(d_{\rm data})}$.
\item compare $\rho_{\rm data}^0$ with $\tilde\rho_{\rm model}^0= \tilde \mu_{\rm model}^{\rm out}  \tilde N_{\rm model}$.
\end{enumerate}

\subsubsection{Application to a real prison}

In this part, we apply the procedure described in the previous section to estimate the total number of inmates, entry, removal and exit rates, using real data from the prison \emph{Colina 1} in the Metropolitan Region (Santiago), Chile. This prison is a penitentiary center in closed regime where convicted inmates, eventually transferred from other prisons, serve their sentences.

To estimate the distribution $(\pi^i)_i$ we consider the average of histograms of the initial sentencing lengths of the existing inmate population in the whole Metropolitan Region. Figure \ref{fig:histoRM} shows the proportions by class associated to these histograms for the years 2016-2019, obtained from \cite{histograms}, which are quite similar among different years. In these histograms, the frequencies are distributed in the following intervals of initial sentencing lengths: $(0, 15]$ days, $(15,600]$ days, $(600$ days$,$ $3$ years$]$, $(3,5]$ years, $(5,10]$ years, $(10,15]$ years, $(15,20]$ years, and $(20,40]$ years. The average vector obtained is $\pi_{\rm data}=(0.002,0.14, 0.061, 0.219, 0.341, 0.153, 0.049, 0.035)$.

\begin{figure}[ht]
\begin{center}
\includegraphics[scale=0.6]{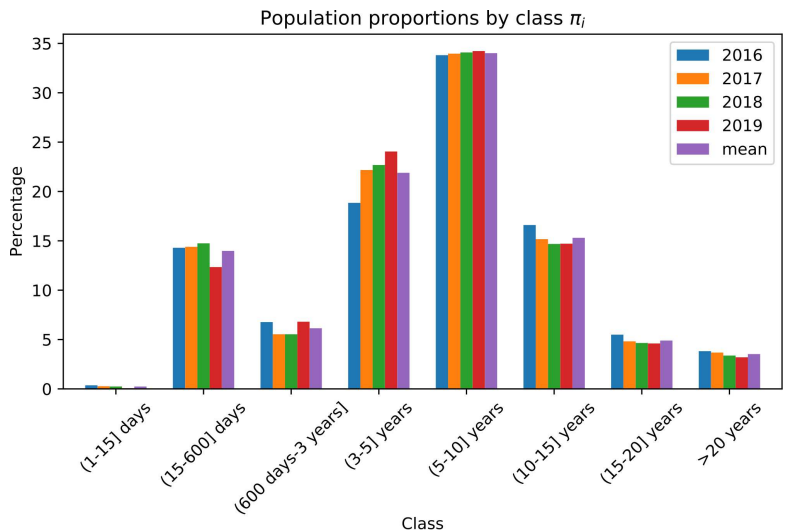}
\captionof{figure}{Histograms of the initial sentencing lengths of the existing inmate population in the whole Metropolitan Region for the years 2016-2019.}\label{fig:histoRM}
\end{center}
\end{figure}

Regarding the estimation of the other parameters, we had access to the data from  prison \emph{Colina 1} for years 2010-2014 and 2019-2020 in \cite{prisonreports}. The data we use is presented in Table \ref{table:colina1}. This information includes: 
\begin{itemize}
\item monthly averages of the number of total inmates $N_{\rm data}$ by year;
\item yearly entry rate $\lambda_{\rm data}$ of new inmates or transferred from other prisons; 
\item yearly number $D_{\rm data}$ of removed inmates due to deaths, transfers to other prisons, etc.;
\item yearly exit rate $\rho_{\rm data}^0$ of inmates who are finishing their sentences.
\end{itemize}

\begin{table}[t]
\centering
\caption{Data obtained from \cite{prisonreports} corresponding to prison \emph{Colina 1} (Santiago, Metropolitan Region, Chile): Number of total inmates ($N_{\rm data}$), total entry rate $\lambda_{\rm data}$ of new inmates or transferred from other prisons; the removal rate $d_{\rm data}$ of inmates due to deaths, transfers to other prisons, etc., and the exit rate $\rho_{\rm data}^0$ of inmates who are finishing their sentences.} \label{table:colina1}
\begin{tabular}{llllll}
  \hline \\  [-0.5ex]
  \textbf{Year} & $N_{\rm data}$ &$\lambda_{\rm data}$ & $\rho_{\rm data}^0$ & $D_{\rm data}$ & $d_{\rm data}$ \\ [1.5ex]
  \hline \\ [-1.5ex]
2010 & 1,838 & 1,086 & 404 & 289 & 0.16\\
2011 & 1,986 & 1,086 & 352 & 363 & 0.18\\ 
2012 & 2,152 &   539 & 396 & 523 & 0.24\\
2013 & 1,674 &   802 & 251 & 613 & 0.37\\
2014 & 1,752 &   804 & 254 & 736 & 0.42\\
2019 & 2,011 & 1,196 & 110 & 931 & 0.46\\
2020 & 1,951 &   964 &  91 & 894 & 0.46\\[0.5ex]
\hline
\textbf{Averages} &1,909 & 925& 265 & 621 & 0.33\\[0.5ex]
\hline
\end{tabular}
\end{table}

We estimate each of the parameters listed above by the average of the corresponding variables in Table \ref{table:colina1}. The removal rate $d_{\rm data}$ is estimated as the average of the ratios $D_{\rm data}/N_{\rm data}$.

Once obtained the proportions $\pi_{\rm data}^i$, to estimate the probabilities $\displaystyle \tilde p^i $ we consider as representative $S_i$ of the sentencing length interval $(T_{i-1},T_i]$ the value $S_i=T_i$. The obtained distribution of $\pi_{\rm data}^i$ and  $\displaystyle \tilde p^i $  (and then $q(\cdot)$) are depicted in Figure \ref{fig:distributions}.

\begin{figure}[ht]
\begin{center}
\includegraphics[scale=0.6]{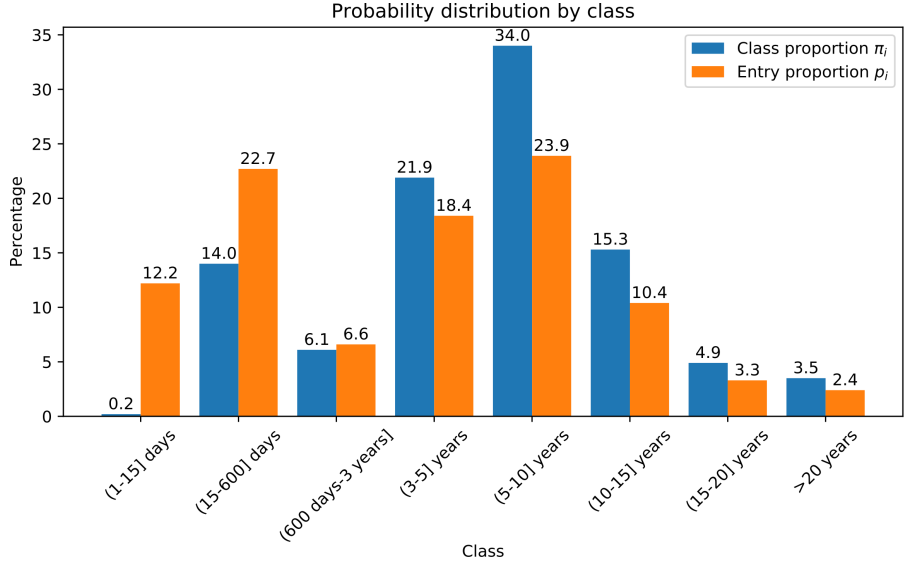}
\captionof{figure}{Distribution of $\pi_{\rm data}^i$ obtained from histograms of the initial sentencing lengths of the existing inmate population in the all Metropolitan Region for the years 2016-2019 (see Fig. \ref{fig:histoRM})
and  estimated distribution $q(\cdot)$ given by estimated probabilities $\displaystyle \tilde p^i$. }\label{fig:distributions}
\end{center}
\end{figure}

The comparison between the total number of inmates ($N$), the entry and exit rates relative to the populations sizes ($\mu^{{\rm in}}$ and $\mu^{{\rm out}}$ respectively), obtained directly from the data and obtained from the model (using $\lambda_{\rm data}$, $d_{\rm data}$ and the histograms), is shown in Table \ref{table:comparison}.

\begin{table}[h!]
\centering
\caption{Comparison between the total number of inmates ($N$), the entry and exit rates relative to the populations sizes ($\mu^{{\rm in}}$ and $\mu^{{\rm out}}$ respectively) obtained directly from the data and obtained from the model (using $N_{\rm data}$, $\lambda_{\rm data}$, $d_{\rm data}$ and the histograms).} \label{table:comparison}
\resizebox{\textwidth}{!}{
\begin{tabular}{lclcc}
  \hline \\  [-0.5ex]
  \textbf{Data} &  \textbf{Value}  &  \textbf{Model}  &  \textbf{Value} &  \textbf{\% Relative difference}   \\ [1.5ex]
  \hline \\ [-1.5ex]
  $N_{\rm data}$ & 1,909&$\tilde N_{\rm model} = \lambda_{\rm data}\frac{1-\hat{q}(d_{\rm data})}{d_{\rm data}}$&1,914& 0.26\%\\ [1.5ex]
$\mu_{\rm data}^{\rm in}=\lambda_{\rm data}/N_{\rm data}$ & 0.49 &$\tilde \mu_{\rm model}^{\rm in} = \frac{d_{\rm data}}{1-\hat q(d_{\rm data})}$& 0.48 & -0.49\% \\ [1.5ex]
$\mu_{\rm data}^{\rm out}=\rho_{\rm data}^0/N_{\rm data}$ & 0.14&  $\tilde \mu_{\rm model}^{\rm out} = \frac{d_{\rm data}\hat q(d_{\rm data})}{1-\hat q(d_{\rm data})}$&0.16& 14.28\%\\ [1.5ex]
$\rho_{\rm data}^0$ &265& $\tilde \rho_{\rm model}^0= \tilde \mu_{\rm model}^{\rm out}  \tilde N_{\rm model}$& 299& 12.83\%\\ [1.5ex]
  \hline
\end{tabular}
}
\end{table}

The differences observed in the Table \ref{table:comparison} are not very large and they can be explained by: (i) The consideration of the histograms of the initial sentencing lengths of the existing inmate population in all the Metropolitan Region and not just those of the prison studied (not available); (ii) The assumption that all the admitted inmates start serving their sentence in the prison under study, which is not totally true for inmates transferred from other prisons where they have served part of their sentences; (iii) The consideration of the value $T_i$ as representative of the sentence length interval $(T_{i-1},T_i]$; (iv) The assumption that the removal rate $d$ is independent of the remaining sentence time in prison, which, depending on the reason of the removal (for instance, by pardon or commutation of the sentence) can be a too strong hypothesis. Nevertheless, the results obtained with the model suggest a reliable approximation of the analyzed prison situation.


\section{SIS models in inmate populations}\label{sec:SIS}

In this section, we consider the spread of a communicable disease in a prison population, modeling the disease dynamics by the SIS (susceptible-infected-susceptible) model (i.e., the disease confers no immunity). In general, an SIS model is appropriate for a bacterial disease. Suppose that disease transmission occurs at the per capita contact rate (sufficient to transmit the disease) $\beta_t > 0$, and infective individuals recover from the disease at a rate $\gamma_t >0$ with no immunity. Assume further that removals from the prison occur at a rate $d_t\geq0$, the entry rate of inmates to the prison is $\lambda_t\geq0$, and the distribution of initial sentence lengths is $q_t(\cdot)$ (as considered in Section \ref{sec:modelo}). We assume that a proportion $\alpha_I>0$ of the new inmates is infective and that a proportion $\alpha_S = 1-\alpha_I$ is susceptible to the infection. The proportion $\alpha_I>0$ is assumed to be the steady-state prevalence of the population outside the prison.

First, we consider the whole prison as a single class. Define $S_t \geq 0$ and $I_t \geq 0$ as the quantities of susceptible and infective individuals, respectively, and $N_t$ as the population size at time $t\geq0$. Then, the disease transmission dynamics in the population can be described by the following system of differential equations:
\begin{equation}\label{eq:sys_1clase}
\left\{\quad\begin{split}
\dot{N}_t & = \lambda_t - \rho^0_t-d_t N_t,\\
\dot{S}_t & = \alpha_S\lambda_t - \beta_t S_t\frac{I_t}{N_t} + \gamma_t I_t -\rho^0_t \frac{S_t}{N_t} - d_t S_t, \\
\dot{I}_t & = \alpha_I\lambda_t + \beta_t S_t\frac{I_t}{N_t} - \gamma_t I_t - \rho^0_t \frac{I_t}{N_t} - d_t I_t,
\end{split}\right.
\end{equation}
where the first equation comes from \eqref{eq:balance_poblacion}, with $\rho^0_t$ being the instantaneous exit rate of prisoners, which is a solution of \eqref{eq:pde_muerte}. We note that the set $\{(N,S,I)\,|\,N-(S+I)=0\}$ is invariant under \eqref{eq:sys_1clase}. Indeed,
\begin{equation*}
\frac{d}{dt}(N_t-(S_t+I_t)) = -\left[ \frac{\rho^0_t}{N_t} + d_t \right](N_t-(S_t+I_t)).
\end{equation*}
Since we suppose that, at the beginning of the process, $N_0=S_0+I_0$, we can replace $S_t=N_t-I_t$ in the equation for $I$, and defining $x_t = I_t/N_t$, we obtain the equation
\begin{equation}\label{eq:SIS-logistic}
\dot x_t 
\,=\, \alpha_I\mu_t^{{\rm in}} + \beta_t x_t(1-x_t) - \gamma_t x_t - \mu_t^{{\rm in}} x_t, \qquad \mu_t^{{\rm in}}:= \frac{\lambda_t}{N_t}.
\end{equation}

Now, suppose that the population is divided into two classes of initial sentence lengths, as described in Section \ref{subsec:clases}, where $i=1$ (resp. $i=2$) stands for the class of short (resp. long) sentences, without affecting the homogeneous mixing of inmates. Define $S_t^i \geq 0$ and $I_t^i \geq 0$ as the quantities of susceptible and infective individuals that belong to class $i$, respectively, with $N_t^i$ being the population size of class $i$ at time $t\geq0$ ($i=1,2$). Each susceptible individual of a class may have contact with an infective individual of his/her own class or of the other class. We suppose that both classes share the same removal rate $d_t\geq0$. Then, the disease transmission dynamics can be described by the following system of coupled differential equations for $i=1,2$:
\begin{equation}\label{eq:sys_2clases}
\left\{\quad\begin{split}
\dot N_t^i & = \lambda^i_t - \rho^{0,i}_t-d_t N_t^i,\\
\dot S_t^i & = \alpha_S\lambda^i_t - \beta_t S_t^i\frac{I_t^1+I_t^2}{N_t} + \gamma_t I_t^i -\rho^{0,i}_t \frac{S_t^i}{N_t^i} - d_t S_t^i, \\
\dot I_t^i & = \alpha_I\lambda^i_t + \beta_t S_t^i\frac{I_t^1+I_t^2}{N_t} - \gamma_t I_t^i - \rho^{0,i}_t \frac{I_t^i}{N_t^i} - d_t I_t^i,
\end{split}\right.
\end{equation}
where the equations for $N^i$ come from \eqref{eq:balance_poblacioN^i} and $\rho_t^{0,i}$ are solutions of \eqref{eq:pde_muerte_i}.
In this case, we also have the invariance of the set $\{(N^i,S^i,I^i)\,|\, N^i-(S^i+I^i) =0,\,i=1,2 \}$, and then we can replace $S_t^i=N_t^i-I_t^i$ in the equation for $I^i$ in \eqref{eq:sys_2clases}. Thus, $I^i$ solves
\begin{equation*}
\dot I_t^i = \alpha_I\lambda^i_t + \beta_t\frac{I_t^1+I_t^2}{N_t^1+N_t^2}(N_t^i-I_t^i) - \gamma_t I_t^i - \rho_t^{0,i}\frac{I_t^i}{N_t^i} - d_t I_t^i.
\end{equation*}

Thanks to Lemma \ref{lemma:N_Ni}, $N=N^1+N^2$ solves the same equation as in \eqref{eq:sys_1clase}. Defining $x_t^i:= I_t^i/N_t$ as the proportion of infective people in class $i$ with respect to the total prison population, $x^i$ satisfies
\begin{equation}\label{eq:SIS2class-logistic}
\dot x_t^i 
\,=\, \alpha_I\mu_t^{{\rm in},i}\pi_t^i + \beta_t(x_t^1+x_t^2)(\pi_t^i-x_t^i) - \gamma_t x_t^i - \mu_t^{{\rm out},i}x_t^i - (\mu_t^{\rm in}-\mu_t^{\rm out}) x_t^i,\quad i=1,2,
\end{equation}
where $\pi_t^i=N_t^i/N_t$ is the proportion of inmates of each class relative to the total prison population at each time $t$ and $\mu_t^{{\rm in}}$, $\mu_t^{{\rm out}}$, $\mu_t^{{\rm in},i}$, and $\mu_t^{{\rm out},i}$ as in \eqref{eq:tasas_t}.

\subsection{Steady-state prevalence comparison between SIS models with and without sentencing length structure}

In this section, we analyze and compare the equilibria of equations \eqref{eq:SIS-logistic} (single-class model) and \eqref{eq:SIS2class-logistic} (two-class model) and provide conditions under which the single-class model underestimates the proportion of infected inmates with respect to the two-class model.

\begin{remark}
Note that the equations associated with the populations $N,N^1,N^2$ in models \eqref{eq:sys_1clase} and \eqref{eq:sys_2clases} are independent of the epidemiological partition $S,I$, since we do not consider specific removals due to illness. Thus, if $\lambda$, $d$ and $q(\cdot)$ do not depend on $t$, we can perform a partial analysis considering the populations $N,N^1,N^2$ in equilibrium. Then, the total populations $N$, $N^1$, $N^2$ are constant (given by \eqref{eq:N^i_muerte} and \eqref{eq:N_muerte}), as are the rates $\mu^{{\rm in}}$, $\mu^{{\rm out}}$, $\muini$, $\muouti$ and the proportions $\pi^i$, $i=1,2$ which, from \eqref{eq:tasas_t}, become
\begin{equation}\label{eq:tasas_eq}
\begin{split}
\mu^{{\rm in}} = \frac{d}{1-\hat q(d)},\quad
\mu^{{\rm out}} = \hat q(d)\mu^{{\rm in}},\,\\[2mm]
\muini = \frac{d}{1-\hat q^i(d)},\quad
\muouti = \hat q^i(d)\muini,\quad
\pi^i = p^i\frac{1-\hat{q}^i(d)}{1-\hat q(d)}, 
\end{split}
\end{equation}
As we suppose the same (constant) removal rate for both classes, \eqref{eq:muin_muout_delta} states that
\begin{equation*}\label{eq:aux0003}
\mu^{{\rm in}} - \mu^{{\rm out}} \,=\, d\,=\, \muini - \muouti, \quad i=1,2.
\end{equation*}

Then, supposing that $\beta,\gamma$ does not depend on $t$, the equations corresponding to the epidemiological parts of the models \eqref{eq:SIS-logistic} and \eqref{eq:SIS2class-logistic}, under the assumption of population at equilibrium, take the simpler form
\begin{align}\label{eq:SIS-logistic-x}
\dot x_t =&\, \alpha_I\mu^{{\rm in}} + \beta x_t(1-x_t) - (\gamma + \mu^{{\rm in}}) x_t,\\ \label{eq:SIS-logistic-xi} 
\dot x_t^i =&\, \alpha_I\pi^i\muini + \beta(x_t^1+x_t^2)(\pi^i-x_t^i) - (\gamma + \muini)x_t^i,\quad i=1,2.
\end{align}

Since in this section we study the behavior in equilibrium, it suffices to study the equilibria of \eqref{eq:SIS-logistic-x} and \eqref{eq:SIS-logistic-xi}.
\end{remark}

Let us denote the positive equilibrium of the single-class model
\eqref{eq:SIS-logistic-x} by $\xeq$, and the positive equilibrium of the two-class model
\eqref{eq:SIS-logistic-xi} by $(\xeqa,\xeqb)$. Define the total proportion of infective individuals in equilibrium in the two-class model by $\omegaeq := \xeqa+\xeqb$.
\medskip

We present the main result of this section in Proposition \ref{prop:comparacion_clases_eq}:

\begin{proposition}\label{prop:comparacion_clases_eq}
Suppose that $\alpha_I,\muini>0$, $i=1,2$. Then, $\xeq,\omegaeq>0$. Moreover, $\xeq<\omegaeq$ ($>,=$ resp.) if and only if $\alpha_I<1-\frac{\gamma}{\beta}$ ($>,=$ resp. ).
\end{proposition}

\begin{proof}
For the single-class model, from \eqref{eq:SIS-logistic-x}, we have the equilibrium equation
\begin{equation*}\label{eq:SIS-logistic_eq}
\alpha_I\mu^{{\rm in}} + \beta \xeq(1-\xeq) - \gamma \xeq - \mu^{{\rm in}} \xeq \,=\, 0  ,\qquad \mu^{{\rm in}} \,=\,  \frac{d}{1-\hat q(d)},
\end{equation*}
with $\mu^{{\rm in}}$ given in \eqref{eq:tasas_eq}. Then, $\xeq$ satisfies
\begin{equation*}\label{eq:xeq_eq}
\beta(\xeq)^2 - (\beta- (\gamma + \mu^{{\rm in}}))\xeq -\alpha_I\mu^{{\rm in}} \,=\, 0,
\end{equation*}
from which we obtain the alternative equation
\begin{equation}\label{eq:xeq_eq_alt}
(\xeq)^2 \,=\, \frac{1}{\beta} \left[(\beta- (\gamma + \mu^{{\rm in}}))\xeq -\alpha_I\mu^{{\rm in}}\right],
\end{equation}
and the explicit expression for the nonnegative solution
\begin{equation}\label{eq:xeq_sol}
\xeq \,=\, \frac{1}{2\beta}\left[ (\beta- (\gamma + \mu^{{\rm in}})) +\sqrt{(\beta- (\gamma + \mu^{{\rm in}}))^2 + 4\alpha_I\mu^{{\rm in}}\beta} \right],
\end{equation}
which is strictly positive if $\alpha_I\mu^{{\rm in}}>0$.

For the two-class model, from \eqref{eq:SIS-logistic-xi}, we have the equilibrium equations
\begin{equation}\label{eq:SIS2class-logistic_eq}
\left\{\quad\begin{split}
\alpha_I\pi^1\mu^{{\rm in},1} + \beta(\xeqa+\xeqb)(\pi^1-\xeqa) - \gamma \xeqa - \mu^{{\rm in},1}\xeqa  = &\,0, \\
\alpha_I\pi^2\mu^{{\rm in},2} + \beta(\xeqa+\xeqb)(\pi^2-\xeqb) - \gamma \xeqb - \mu^{{\rm in},2}\xeqb = &\,0,
\end{split}\right.
\end{equation}
with $\pi^i$ and $\muini$ given by \eqref{eq:tasas_eq}. From \eqref{eq:SIS2class-logistic_eq}, we obtain
\begin{equation}\label{eq:xeqab_eq_alt0}
\beta(\xeqa+\xeqb) = \frac{ (\gamma + \mu^{{\rm in},1})\xeqa - \alpha_I\pi^1\mu^{{\rm in},1}}{\pi^1-\xeqa} = \frac{ (\gamma + \mu^{{\rm in},2})\xeqb - \alpha_I\pi^2\mu^{{\rm in},2}}{\pi^2-\xeqb}.
\end{equation}

From \eqref{eq:xeqab_eq_alt0}, we can write $\xeqa$ and $\xeqb$ as functions of $\omegaeq=\xeqa+\xeqb$ as
\begin{equation}\label{eq:xeqab_omega}
\xeqa = g_1(\omegaeq) := \pi^1\frac{ \beta\omegaeq + \alpha_I\mu^{{\rm in},1} }{\beta\omegaeq+\gamma+\mu^{{\rm in},1}},\quad \xeqb = g_2(\omegaeq):=\pi^2\frac{ \beta\omegaeq + \alpha_I\mu^{{\rm in},2}  }{\beta\omegaeq+\gamma+\mu^{{\rm in},2}}.
\end{equation}

Thus, summing both expressions in \eqref{eq:xeqab_omega}, $\omegaeq$ is a nonnegative solution of
\begin{equation}\label{eq:aux004}
\omegaeq = g_1(\omegaeq) + g_2(\omegaeq).
\end{equation}

We contend that, under the hypotheses $\alpha_I,\muini>0$, $i=1,2$, there exists a unique strictly positive solution of \eqref{eq:aux004}. Indeed, each of the functions $g_i(\cdot)$ has the form
\begin{equation*}
g_i(\omega) = \frac{a_i\omega+b_i}{c_i\omega+d_i},
\end{equation*}
with $a_i=\beta\pi^i>0$, $b_i=\alpha_I\pi^i\muini>0$, $c_i=\beta>0$, $d_i=\gamma+\muini>0$. Then, $g_i(\cdot)$ has a unique zero at $\omega_0^i=-b_i/a_i=-\alpha_I\muini/\beta<0$, it is undefined at $\omega_{\infty}^i=-d_i/c_i=-(\gamma+\muini)/\beta<0$, and it holds that $\omega_{\infty}<\omega_0$. Indeed, this is equivalent to $a_id_i-b_ic_i=\beta\pi^i(\gamma+(1-\alpha_I)\muini)>0$. Moreover,
\begin{equation*}
g_i'(\omega) = \frac{a_id_i-b_ic_i}{(c_i\omega+d_i)^2}
\end{equation*}
which is strictly positive for $\omega>\omega_{\infty}^i$ and decreases to 0 as $\omega\rightarrow\infty$. This shows that $g_i(\cdot)$ is strictly increasing and concave on the interval $(\omega_{\infty}^i,\infty)$ and positive on the interval $(\omega_{0}^i,\infty)\subseteq(\omega_{\infty}^i,\infty)$, with $g_i(0)=b_i/d_i>0$ and $\lim_{\omega\rightarrow\infty}g_i(\omega)=a_i/c_i=\pi^i>0$.
\medskip

Now, consider the function $g_{1,2}(\omega)=g_1(\omega)+g_2(\omega)$. This function is strictly increasing in the interval $(\max\{\omega_{\infty}^1,\omega_{\infty}^2\},\infty)$, with $g_{1,2}(0)>0$, and it has a horizontal asymptote as $\omega$ converges to infinity. Thus, there exists a unique strictly positive solution of equation $g_{1,2}(\omega)=\omega$, that is, of \eqref{eq:aux004}.

We refer to the unique positive solution of \eqref{eq:aux004} as $\omegaeq$, which can be equivalently written (under a rearrangement of the terms of \eqref{eq:aux004}) as the unique positive root of the function
\begin{equation*}\label{eq:omega_eq}
\begin{split}
g(\omega)\,:=\,& \beta^2\omega^3 + \beta(\mu^{{\rm in},1}+\mu^{{\rm in},2}+2\gamma-\beta)\omega^2  +((\gamma+\mu^{{\rm in},1})(\gamma+\mu^{{\rm in},2})\\[2mm]
&-\pi^1\beta(\alpha_I\mu^{{\rm in},1}+\gamma+\mu^{{\rm in},2})-\pi^2\beta(\alpha_I\mu^{{\rm in},2}+\gamma+\mu^{{\rm in},1}))\omega \\[2mm]
& - \alpha_I(\pi^1\mu^{{\rm in},1}(\gamma+\mu^{{\rm in},2}) + \pi^2\mu^{{\rm in},2}(\gamma+\mu^{{\rm in},1})  ).
\end{split}
\end{equation*}

The function $g(\cdot)$ is a third-degree polynomial, with $\lim_{\omega\rightarrow\infty}g(\omega)=\infty$, $g(0)<0$, and $g(\omegaeq)=0$. Thus, on the interval $[0,\infty)$, $g(\omega)<0$ if and only if $\omega<\omegaeq$, and $g(\omega)>0$ if and only if $\omega>\omegaeq$. To compare $\xeq$ and $\omegaeq$, it suffices to compute the sign of $g(\xeq)$, provided that $\xeq\geq0$. Using \eqref{eq:xeq_eq_alt} and \eqref{eq:aux005}, after a lengthy computation, we arrive at
\begin{equation*}
g(\xeq) = (\mu^{{\rm in}}-\mu^{{\rm in},1})(\mu^{{\rm in}}-\mu^{{\rm in},2})(\xeq-\alpha_I), 
\end{equation*}
where, from Corollary \ref{remark:orden_muis}, $\mu^{{\rm in},2}\leq\mu^{{\rm in}}\leq\mu^{{\rm in},1}$. Then,
\begin{equation*}
\xeq<\omegaeq~(\mbox{resp.}>,=)\quad \Leftrightarrow \quad g(\xeq)<0~(\mbox{resp.}>,=) \quad\Leftrightarrow\quad \xeq>\alpha_I~(\mbox{resp.}<,=).
\end{equation*}

Using the formula for $\xeq$ from \eqref{eq:xeq_sol}, we obtain the condition for the equilibrium $\xeq$ to be less than (resp. greater than, equal to) $\omegaeq$:
\begin{equation*}
\xeq>\alpha_I ~(\mbox{resp.}<,=) \quad\Leftrightarrow\quad 1-\frac{\gamma}{\beta}>\alpha_I ~(\mbox{resp.}<,=),
\end{equation*}
which concludes the proof.
\end{proof}

\begin{remark}
From Proposition \ref{prop:comparacion_clases_eq}, the condition $1-\frac{\gamma}{\beta}>\alpha_I$ for obtaining $\xeq<\omegaeq$ does not depend on the removal rate $d$ or on the parameters of the class separation. On the other hand, the threshold $1-\frac{\gamma}{\beta}$ is exactly the herd immunity threshold (i.e., $1 -1/\mathcal{R}_0$ with $\mathcal{R}_0=\beta/\gamma$) associated with the single-class model \eqref{eq:SIS-logistic-x} when $\mu^{\rm in} = 0$.
\end{remark}

\subsection{Numerical simulations}\label{subsec:simulations}

For the numerical simulations, we consider the epidemiological parameters $\beta=0.5$ and $\gamma=0.1$ and four cases of removal rates: null removal rate ($d=0$) and positive removal rates $d=0.1,\, 0.2, \, 0.3$. We consider $r^{\star}=5$ [years] as the maximum time for a sentence to be considered \emph{short}. Under these parameters, Proposition \ref{prop:comparacion_clases_eq} states that $\xeq<\omegaeq$ if $\alpha_I<1-\frac{\gamma}{\beta}=0.8$.

We study the cases of initial sentence lengths given by an exponential function
\begin{equation*}
q'(r) = \frac{1}{10}e^{-\frac{r}{10}},\quad r\geq0,
\end{equation*}
the mean sentence time of which is $T_q=10$ [years], and by a bimodal function
\begin{equation*}
q'(r) = 0.6\frac{1}{p_{5,1^2}}\phi(r;5,1^2) + 0.4\frac{1}{p_{10,1.5^2}}\phi(r;10,1.5^2), \quad r\geq0
\end{equation*}
where $\phi(\cdot;\mu,\sigma^2)$ denotes the probability density function of a normal distribution with mean $\mu$ and variance $\sigma^2$, and $p_{\mu,\sigma^2}=\int_0^{\infty}\phi(t;\mu,\sigma^2)dt$ is a normalizing constant, having the mean sentence time $T_q=7$ [years].

In the first example (exponential distribution), we obtain the proportions $p^1=0.39$ and $p^2=0.61$ and the mean initial sentence lengths by class $T_{q^1}=2.29$ [years] and $T_{q^2}=15$ [years]. In the second example, we have the proportions $p^1=0.30$ and $p^2=0.70$ and the mean initial sentence lengths by class $T_{q^1}=4.20$ [years] and $T_{q^2}=8.20$ [years].

In Figures \ref{fig:comp_eqs_rel_exp} (exponential) and \ref{fig:comp_eqs_rel_bimod} (bimodal), we compute the underestimation of the infected proportion of the population incurred by the single-class model with respect to the two-class model relative to the disease prevalence of new inmates $\alpha_I$, that is, $(\omegaeq-\xeq)/\alpha_I$, and relative to the single-class prevalence $\xeq$, that is, $(\omegaeq-\xeq)/\xeq$.

We show the plots for prevalences in the interval $\alpha_I\in [0,0.1]$, which is a range for $\alpha_I$ in which $\omegaeq>\xeq$, that is, the single-class model underestimates the number of infected people in the prison, according to Proposition \ref{prop:comparacion_clases_eq}.

\begin{figure}[h]
\begin{center}
\includegraphics[scale=0.6]{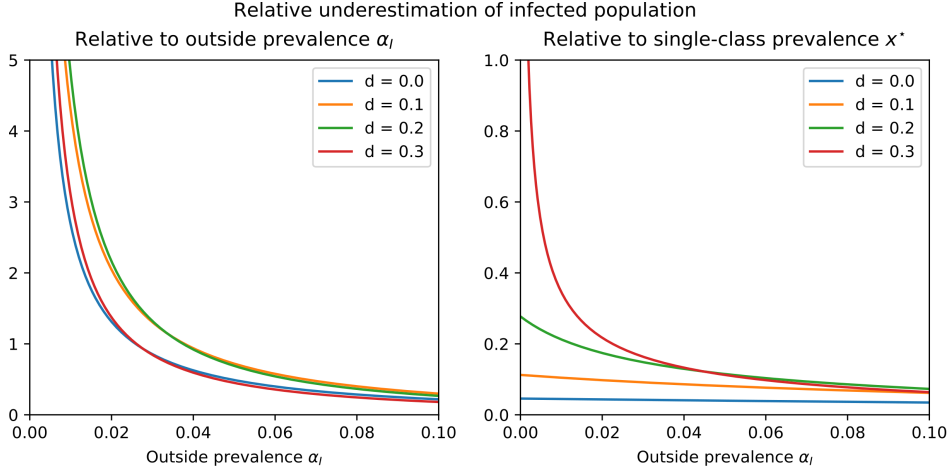}
\captionof{figure}{Underestimation of total infected proportion by the single-class model with respect to the two-class model, relative to outside prevalence (left) and relative to the single-class model prevalence (right) when considering an exponential distribution of initial sentencing lengths. }\label{fig:comp_eqs_rel_exp}
\end{center}
\end{figure}

\begin{figure}[h]
\begin{center}
\includegraphics[scale=0.6]{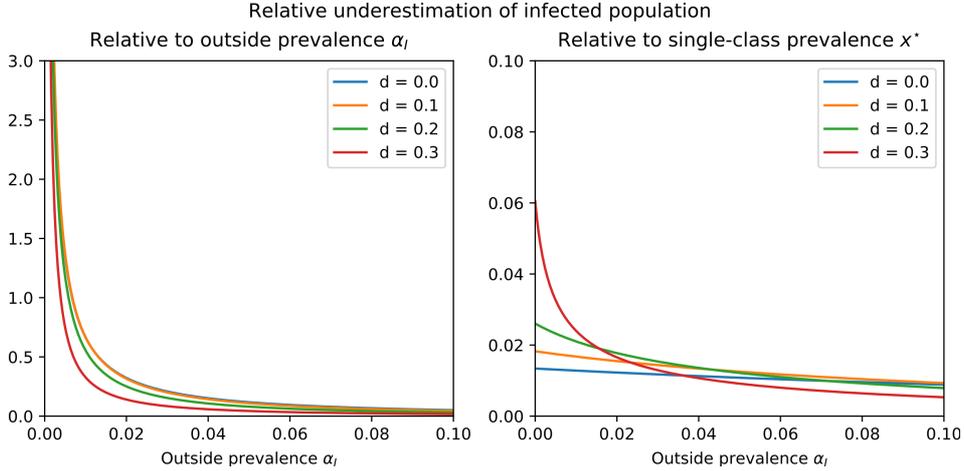}
\captionof{figure}{Underestimation of total infected proportion by the one-class model with respect to the two-class model, relative to outside prevalence (left) and relative to the single-class model prevalence (right) when considering a bimodal distribution of initial sentencing lengths. }\label{fig:comp_eqs_rel_bimod}
\end{center}
\end{figure}

For the second example (bimodal distribution), we depict in Figure \ref{fig:comp_dists_bimod} a comparison between the initial sentence length distribution and the remaining sentence time distribution. Note that the remaining sentence time distribution inside the prison approaches the initial sentence length as $d$ increases.

\begin{figure}[ht]
\begin{center}
\includegraphics[scale=0.6]{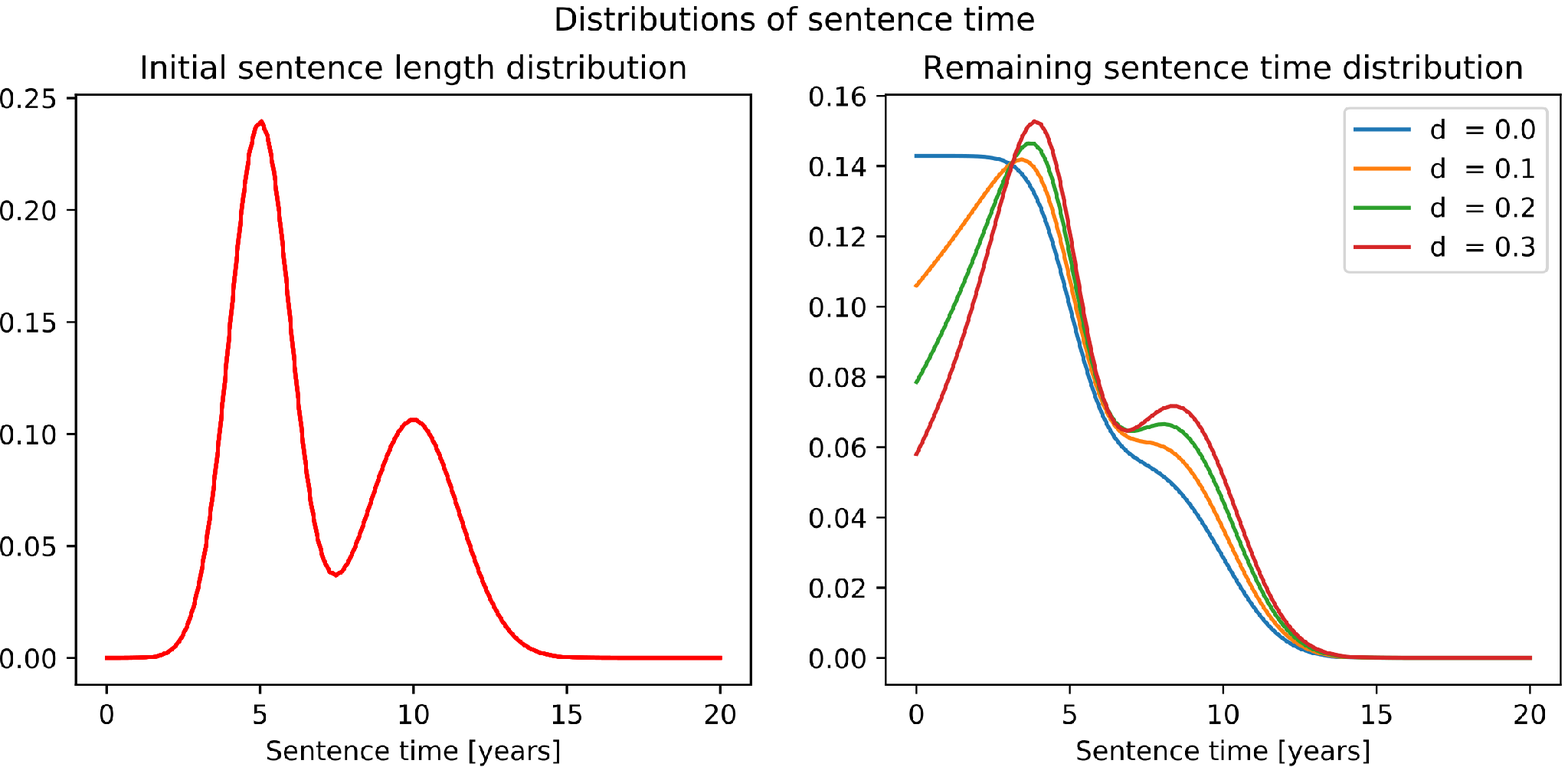}
\captionof{figure}{Comparison of initial sentence length distribution (bimodal) and remaining sentence time distribution inside the prison for different values of $d$. }\label{fig:comp_dists_bimod}
\end{center}
\end{figure}

\section{Conclusions}\label{sec:conclusion}

In this paper, we introduce an inmate population model with a sentencing length structure and find that the density of the number of inmates (with respect to the remaining time in prison) follows a transport equation, typically known as the McKendrick equation. We compute the inmate population and the entry/exit rates at steady state, showing that the sentencing length structure and the removal rate have a strong influence on these values. We illustrate how to obtain these values with real data from a prison in Chile. Since a typical assumption in prison models is a constant prison population, the obtained values can be used by decision-makers in the design or optimization of a penitentiary system.

To study the effect of considering or not considering sentence length structure in the estimation of the infected population, we divide the inmate population into two classes depending on their sentencing lengths (short and long), and we couple the SIS epidemiological model to the obtained model. This epidemiological model is compared, in equilibrium, with the model obtained when ignoring the sentencing length structure. We prove that not accounting for the structure of the sentence lengths for disease prevalences of new inmates below a certain threshold induces an underestimation of the prevalence in the prison population at steady state. The involved threshold depends on the basic reproduction number associated with the nonstructured SIS model with no entry of new inmates.

In epidemiological models for inmate populations, assuming that the prevalence of new inmates is low represents a situation where the disease under study is almost eradicated from the general population (outside the prison), but since the prison can act as a reservoir, the prevalence of this disease in the inmate population can eventually be much higher (see, for instance, \cite{castillo2020,Legrand:2008,Witbooi:2017}). Therefore, the message from our results is that, in these situations, a recommendation from the modeling perspective is to include the sentencing length structure of the prison population when analyzing epidemiological models.


\section{Acknowledgments}

This work was funded by FONDECYT grants N 1200355 (first author) and N 3180367 (second author) from ANID-Chile. The authors are very grateful to professors Heliana Arias (Universidad del Valle, Colombia) and Carla Castillo-Laborde (Universidad del Desarrollo, Chile) for fruitful discussions.


\bibliographystyle{siamplain}
\bibliography{preprint_inhomo_sis_prison_2021-01-29}

\begin{thebibliography}{10}

\bibitem{Akushevich:2019}
{\sc I.~Akushevich, A.~Yashkin, J.~Kravchenko, F.~Fang, K.~Arbeev, F.~Sloan,
  and A.~I. Yashin}, {\em A forecasting model of disease prevalence based on
  the {M}c{K}endrick--von {F}oerster equation}, Math. Biosci., 311 (2019),
  pp.~31--38, \url{https://doi.org/10.1016/j.mbs.2018.12.017}.

\bibitem{Ayer:2019}
{\sc T.~Ayer, C.~Zhang, A.~Bonifonte, A.~C. Spaulding, and J.~Chhatwal}, {\em
  Prioritizing hepatitis {C} treatment in {U}.{S}. prisons}, Oper. Res., 67
  (2019), pp.~853--873, \url{https://doi.org/10.1287/opre.2018.1812}.

\bibitem{Beauparlant:2016}
{\sc M.~Beauparlant and R.~Smith}, {\em A metapopulation model for the spread
  of {MRSA} in correctional facilities}, Infectious Disease Modelling, 1
  (2016), pp.~11 -- 27,
  \url{https://doi.org/https://doi.org/10.1016/j.idm.2016.06.001}.

\bibitem{belenko2008}
{\sc S.~Belenko, R.~Dembo, D.~Weiland, M.~Rollie, C.~Salvatore, A.~Hanlon, and
  K.~Childs}, {\em Recently arrested adolescents are at high risk for sexually
  transmitted diseases}, Sexually transmitted diseases, 35 (2008),
  pp.~758--763.

\bibitem{Bog2007}
{\sc V.~I. Bogachev}, {\em Measure theory. {V}ol. {I}, {II}}, Springer-Verlag,
  Berlin, 2007, \url{https://doi.org/10.1007/978-3-540-34514-5}.

\bibitem{BrauerCastillo:2012}
{\sc F.~Brauer and C.~Castillo-Chavez}, {\em Mathematical models in population
  biology and epidemiology}, vol.~40 of Texts in Applied Mathematics, Springer,
  New York, second~ed., 2012, \url{https://doi.org/10.1007/978-1-4614-1686-9}.

\bibitem{Brauer:2001}
{\sc F.~Brauer and P.~van~den Driessche}, {\em Models for transmission of
  disease with immigration of infectives}, Math. Biosci., 171 (2001),
  pp.~143--154, \url{https://doi.org/10.1016/S0025-5564(01)00057-8}.

\bibitem{Busenberg:2012}
{\sc S.~Busenberg and K.~Cooke}, {\em Vertically transmitted diseases: models
  and dynamics}, vol.~53 of Biomathematics, Springer-Verlag, Berlin, 1993.

\bibitem{castillo2020}
{\sc C.~Castillo-Laborde, P.~Gajardo, M.~N{\'a}jera-De~Ferrari, I.~Matute,
  M.~Hirmas-Adauy, P.~Aguirre, H.~Ram{\'\i}rez, D.~Ram{\'i}rez, and
  X.~Aguilera}, {\em Modelling cost-effectiveness of syphilis detection
  strategies in prisoners: Exploratory exercise in a {C}hilean male prison},
  Cost Effectiveness and Resource Allocation, 19 (2021),
  \url{https://doi.org/10.1186/s12962-021-00257-9}.

\bibitem{cohn2013}
{\sc D.~L. Cohn}, {\em Measure theory}, Birkh\"{a}user Advanced Texts: Basler
  Lehrb\"{u}cher. [Birkh\"{a}user Advanced Texts: Basel Textbooks],
  Birkh\"{a}user/Springer, New York, second~ed., 2013,
  \url{https://doi.org/10.1007/978-1-4614-6956-8}.

\bibitem{histograms}
{\sc G.~de~Chile}, {\em {Compendios Estad\'{\i}sticos (Statistical Digest)}},
  2021, \url{https://www.gendarmeria.gob.cl/estadisticas_compendios.html}
  (accessed 2021-01-21).

\bibitem{prisonreports}
{\sc G.~de~Chile}, {\em {Estad\'{\i}stica general (General Statistics)}}, 2021,
  \url{https://www.gendarmeria.gob.cl/estadisticaspp.html} (accessed
  2021-01-21).

\bibitem{Gani:1997}
{\sc J.~Gani, S.~Yakowitz, and M.~Blount}, {\em The spread and quarantine of
  {HIV} infection in a prison system}, SIAM J. Appl. Math., 57 (1997),
  pp.~1510--1530, \url{https://doi.org/10.1137/S0036139995283237}.

\bibitem{who2007}
{\sc A.~Gatherer, R.~J{\"u}rgens, and H.~St{\"o}ver}, {\em Health in prisons: a
  WHO guide to the essentials in prison health}, WHO Regional Office Europe,
  2007.

\bibitem{khan2009incarceration}
{\sc M.~R. Khan, I.~A. Doherty, V.~J. Schoenbach, E.~M. Taylor, M.~W. Epperson,
  and A.~A. Adimora}, {\em Incarceration and high-risk sex partnerships among
  men in the {United States}}, Journal of Urban Health, 86 (2009),
  pp.~584--601.

\bibitem{MariaetAl2011}
{\sc M.~R. Khan, M.~W. Epperson, P.~Mateu-Gelabert, M.~Bolyard, M.~Sandoval,
  and S.~R. Friedman}, {\em Incarceration, sex with an {STI}- or {HIV}-infected
  partner, and infection with an {STI} or {HIV} in {B}ushwick, {B}rooklyn,
  {NY}: {A} social network perspective}, American Journal of Public Health, 101
  (2011), pp.~1110--1117, \url{https://doi.org/10.2105/AJPH.2009.184721}.
\newblock PMID: 21233443.

\bibitem{kouyoumdjian2012}
{\sc F.~Kouyoumdjian, D.~Leto, S.~John, H.~Henein, and S.~Bondy}, {\em A
  systematic review and meta-analysis of the prevalence of chlamydia,
  gonorrhoea and syphilis in incarcerated persons}, International journal of
  STD \& AIDS, 23 (2012), pp.~248--254.

\bibitem{Kuniya:2015}
{\sc T.~Kuniya and R.~Oizumi}, {\em Existence result for an age-structured
  {SIS} epidemic model with spatial diffusion}, Nonlinear Anal. Real World
  Appl., 23 (2015), pp.~196--208,
  \url{https://doi.org/10.1016/j.nonrwa.2014.10.006}.

\bibitem{Legrand:2008}
{\sc J.~Legrand, A.~Sanchez, F.~Le~Pont, L.~Camacho, and B.~Larouze}, {\em
  Modeling the impact of tuberculosis control strategies in highly endemic
  overcrowded prisons}, PLoS One, 3 (2008), p.~e2100.

\bibitem{world2016}
{\sc W.~H. Organization et~al.}, {\em Global health sector strategy on sexually
  transmitted infections 2016-2021: toward ending {STI}s}, tech. report, World
  Health Organization, 2016,
  \url{https://apps.who.int/iris/bitstream/handle/10665/246296/WHO-RHR-16.09-eng.pdf}.

\bibitem{Witbooi:2017}
{\sc P.~Witbooi and S.~M. Vyambwera}, {\em A model of population dynamics of
  {TB} in a prison system and application to {South Africa}}, BMC research
  notes, 10 (2017), pp.~1--8.

\end{thebibliography}


\appendix
\section{Convergence of nonnegative measures}

\begin{proposition}\label{prop:Apendice}
Consider $(\varphi_t)_{t\geq0}$ a family of nonnegative measures on $\mathcal B(\R_+)$. Suppose that this family satisfies the following hypothesis:
\begin{equation}\label{hip:medidas}
\forall t\geq0,\forall A\in \mathcal B(\R_+),\forall \varepsilon>0,\exists\eta\in(0,1):~ |z|,|z'|<\eta\Rightarrow |\varphi_{t+z}(A+z')-\varphi_{t}(A+z')|<\varepsilon.
\end{equation}
Then,
\begin{enumerate}
\item For almost every $t\geq0$, $0\leq s<r$,
\begin{equation*}
\lim_{\Delta t\searrow0}\frac{1}{\Delta t} \int_0^{\Delta t}|\varphi_{t+\Delta t-u}((s+u,r+u])-\varphi_{t}((s+u,r+u])|du = 0.
\end{equation*}
\item For almost every $t\geq0$, $0\leq s<r$,
\begin{equation*}
\lim_{\Delta t\searrow0}\frac{1}{\Delta t}\int_{0}^{\Delta t} \varphi_{t+\Delta t - u}((s+u,t+u])du = \varphi_t((s,r]).
\end{equation*}
\end{enumerate}
\end{proposition}

\begin{proof}
\begin{enumerate}
\item Let $\varepsilon>0$. Take, for $t\geq0$ and $A=(s,r]\in\mathcal B(\R_+)$, $\eta\in(0,1)$ from \eqref{hip:medidas}. Then, there exists $\Delta t<\eta$ such that if $z'=u\leq\Delta t<\eta$, $z=\Delta t-u<\eta$, and then,
\begin{equation*}
\frac{1}{\Delta t}\int_0^{\Delta t}|\varphi_{t+\Delta t-u}((s+u,r+u])-\varphi_{t}((s+u,r+u])|du \leq \frac{1}{\Delta t}\epsilon\Delta t = \epsilon,
\end{equation*}
which proves the result.
\item We have
\begin{equation*}
\begin{split}
\int_{0}^{\Delta t} \varphi_{t+\Delta t - u}((s+u,r+u])du \leq \, \int_{0}^{\Delta t} |\varphi_{t+\Delta t - u}((s+u,r+u])\\
\,-\varphi_{t}((s+u,r+u])|du + \int_{0}^{\Delta t} \varphi_{t}((s+u,r+u])du.
\end{split}
\end{equation*}

The first term on the right-hand side can be bounded by the result in point 1. For the second term,
\begin{equation*}
\begin{split}
\int_{0}^{\Delta t} \varphi_{t}((s+u,r+u])du =&\, \int_{0}^{\Delta t} \varphi_{t}((s,r])du+\int_{0}^{\Delta t} \varphi_{t}((r,r+u])du\\
&\, -\int_{0}^{\Delta t} \varphi_{t}((s,s+u])du
\end{split}
\end{equation*}

The first term on the right-hand side of the previous expression is equal to $\Delta t\cdot \varphi_t((s,r])$. Define, for $t,r$ fixed, the function $f_r(u)=\varphi_t((r,r+u])$. This function is measurable, increasing, and bounded. Then, by \cite[Theorem 5.4.2]{Bog2007}, noting that $f_r(0) = \varphi_t((r,r])=0$,
\begin{equation*}
\begin{split}
\lim_{\Delta t\searrow0}\frac{1}{\Delta t}\int_{0}^{\Delta t} \varphi_{t}((r,r+u])du
= \lim_{\Delta t\searrow0}\frac{1}{\Delta t}\int_{0}^{\Delta t} f_r(u)du = f_r(0)=0,
\end{split}
\end{equation*}
and similarly replacing $r$ by $s$. Thus,
\begin{equation*}
\begin{split}
\lim_{\Delta t\searrow0}\int_{0}^{\Delta t} \varphi_{t}((s+u,r+u])du =&\, \varphi_{t}((s,r]).
\end{split}
\end{equation*}
\end{enumerate}
\end{proof}

\end{document}